%% file: PTF-main.tex
\documentclass[oneside,11pt]{article} % For LaTeX2e

\usepackage{amsmath,amsfonts,amssymb,amsthm,commath}

\usepackage{algorithm,algorithmic}

\usepackage[utf8]{inputenc} % allow utf-8 input
\usepackage[T1]{fontenc}    % use 8-bit T1 fonts
\usepackage{microtype}   % better spacing for pdflatex

\usepackage{nicefrac}       % compact symbols for 1/2, etc.

\usepackage{booktabs,enumitem}

\usepackage{graphicx}

\usepackage{url}
\usepackage{xcolor}

\usepackage{hyperref} 

\usepackage{balance} % balance reference

\usepackage{pifont}

\numberwithin{equation}{section}

\newcommand{\ba}{\boldsymbol{a}}

\newcommand{\calA}{\mathcal{A}}

\newcommand{\calH}{\mathcal{H}}

\newcommand{\calN}{\mathcal{N}}

\newcommand{\calP}{\mathcal{P}}

\newcommand{\R}{\mathbb{R}}

\newcommand{\Rn}{\mathbb{R}^n}

\newcommand{\zeronorm}[1]{\left\lVert #1 \right\rVert_{0}}
\newcommand{\twonorm}[1]{\left\lVert #1 \right\rVert_2}
\newcommand{\onenorm}[1]{\left\lVert #1 \right\rVert_{1}}
\renewcommand{\abs}[1]{\left\lvert #1 \right\rvert}

\newcommand{\fronorm}[1]{\left\lVert #1 \right\rVert_{F}}

\newcommand{\infnorm}[1]{\left\lVert #1 \right\rVert_{\infty}}
\newcommand{\inner}[2]{\langle #1, #2 \rangle}

\newcommand{\supp}{\operatorname{supp}}
\newcommand{\tr}{\operatorname{tr}}

\newcommand{\trans}{^{\top}}

\newcommand{\EXP}{\operatorname{\mathbb{E}}} % Expectation
\DeclareMathOperator*{\Var}{\operatorname{Var}}

\renewcommand{\Pr}{\operatorname{Pr}}

\newcommand{\poly}{\mathrm{poly}}

\newcommand{\sign}{\operatorname{sign}}

\newcommand{\compset}[1]{\overline{#1}}

\newcommand{\ind}[1]{\boldsymbol{1}_{#1}}

\newcommand{\hermite}{\mathrm{He}}

\theoremstyle{plain}

\ifx\theorem\undefined
\newtheorem{theorem}{Theorem}
\fi

\ifx\lemma\undefined
\newtheorem{lemma}[theorem]{Lemma}
\fi

\ifx\proposition\undefined
\newtheorem{proposition}[theorem]{Proposition}
\fi

\ifx\corollary\undefined
\newtheorem{corollary}[theorem]{Corollary}
\fi

\theoremstyle{definition}

\ifx\remark\undefined
\newtheorem{remark}[theorem]{Remark}
\fi

\ifx\definition\undefined
\newtheorem{definition}[theorem]{Definition}
\fi

\ifx\conjecture\undefined
\newtheorem{conjecture}[theorem]{Conjecture}
\fi

\ifx\fact\undefined
\newtheorem{fact}[theorem]{Fact}
\fi

\ifx\claim\undefined
\newtheorem{claim}[theorem]{Claim}
\fi

\ifx\assumption\undefined
\newtheorem{assumption}{Assumption}
\fi

\usepackage{xspace}

\newcommand{\oraclexy}{\mathrm{EX}}

\newcommand{\Chow}{\chi}

\newcommand{\polyclass}{\calP_{n,d}}

\newcommand{\sparseclass}{\calP^2_{n,d,s}}

\newcommand{\sparselinearclass}{\calP^1_{n,d,2k}}

\newcommand{\Ltwonorm}[1]{\left\lVert #1 \right\rVert_{L^2}}

\newcommand{\hardthr}{\mathrm{H}}

\newcommand{\filter}{\textsc{SparseFilter}\xspace}

\newcommand{\Xgamma}{X_{\gamma}}

\newcommand{\vcdim}{\mathrm{VCdim}}

\newcommand{\citet}{\cite}
\title{Attribute-Efficient PAC Learning of Low-Degree Polynomial Threshold Functions with Nasty Noise}
\author{
Shiwei Zeng\\
Stevens Institute of Technology\\
\texttt{szeng4@stevens.edu}
\and
{Jie Shen}\\
Stevens Institute of Technology\\
\texttt{jie.shen@stevens.edu}
}
\begin{document}
\maketitle

%\Jie{I settled all algorithm and proof. Dec 23, 2022 -- Jan 1, 2023}

\input{intro.tex}

\input{setup.tex}

\input{alg.tex}

\input{analysis.tex}

\input{con.tex}

\clearpage
\bibliographystyle{alpha}
\bibliography{../jshen_ref.bib}

\clearpage
\appendix

\input{app.tex}

\input{app-output.tex}

\end{document}

%% file: intro.tex
\begin{abstract}
The concept class of low-degree polynomial threshold functions (PTFs) plays a fundamental role in machine learning. In this paper, we study PAC learning of $K$-sparse degree-$d$ PTFs on $\mathbb{R}^n$, where any such concept depends only on $K$ out of $n$ attributes of the input. Our main contribution is a new algorithm that runs in time $({nd}/{\epsilon})^{O(d)}$ and under the Gaussian marginal distribution, PAC learns the class up to error rate $\epsilon$ with $O(\frac{K^{4d}}{\epsilon^{2d}} \cdot \log^{5d} n)$ samples even when an $\eta \leq O(\epsilon^d)$ fraction of them are corrupted by the nasty noise of Bshouty~et~al.~(2002), possibly the strongest corruption model. Prior to this work, attribute-efficient robust algorithms are established only for the  special case of  sparse homogeneous halfspaces. Our key  ingredients are: 1) a structural result that translates the attribute sparsity to a sparsity pattern of the Chow vector under the basis of Hermite polynomials, and 2) a novel attribute-efficient robust Chow vector estimation algorithm which uses exclusively a restricted Frobenius norm to either certify a good approximation or to validate a sparsity-induced degree-$2d$ polynomial as a filter to detect corrupted samples.
\end{abstract}

\section{Introduction}\label{sec:intro}

A polynomial threshold function (PTF) $f: \Rn \rightarrow \{-1, 1\}$ is of the form $f(x) = \sign( p(x) )$ for some $n$-variate polynomial $p$. The class of low-degree PTFs plays a fundamental role in learning theory owing to its remarkable power for rich representations \cite{mansour1994learning,anthony1999neural,HS07sparsePTF,donnell2014book}. In this paper, we study {\em attribute-efficient} learning of  {degree-$d$} PTFs: if the underlying PTF $f^*$ is promised to depend only on at most $K$ unknown attributes of the input, whether and how can one learn $f^*$ by collecting $O(\poly(K^d, \log n))$ samples  under the classic probably approximately correct (PAC) learning model \cite{valiant1984theory}.

A very special case of the problem, the class of sparse halfspaces (i.e. degree-$1$ PTFs), has been extensively investigated in  machine learning and statistics \cite{littlestone1987learning,blum1990learning,gentile2003robustness,plan2013one}, and a fruitful set of results have been established even under strong noise models \cite{plan2013robust,awasthi2016learning,zhang2018efficient,zhang2020efficient,shen2021attribute}. It, however, turns out that the theoretical and algorithmic understanding of learning sparse degree-$d$ PTFs  fall far behind the linear counterpart. 

% \cite{kearns1988learning,bshouty1998new,cesa1999sample,kalai2005agnostic,klivans2009learning,servedio2012attribute,awasthi2017power}, and recently an attribute-efficient algorithm was developed with near-optimal sample complexity and noise tolerance \cite{shen2021attribute}. 

In the absence of noise, the problem can be cast as solving a linear program by thinking of degree-$d$ PTFs as halfspaces on the space expanded by all monomials of degree at most $d$ \cite{maass1994fast}. However, the problem becomes subtle when samples might be contaminated adversarially. In this work, we consider the nasty noise of \citet{bshouty2002pac}, perhaps the strongest noise model in classification.

\begin{definition}[PAC learning with nasty noise]\label{def:nasty}
Denote by $\calH_{d, K}$ the class of $K$-sparse degree-$d$ PTFs on $\Rn$.
Let $D$ be a distribution on $\Rn$ and $f^* \in \calH_{d, K}$ be the underlying PTF.
A nasty adversary $\oraclexy(\eta)$ takes as input a sample size $N$ requested by the learner, draws $N$ instances independently according to $D$ and annotates them by $f^*$, to form a clean sample set $\bar{S} = \{ ( x_i, f^*(x_i) ) \}_{i=1}^N$. The adversary may then inspect the learning algorithm and uses its unbounded computational power to replace at most an $\eta$ fraction with carefully constructed samples for some $\eta < \frac12$, and returns the corrupted set $\bar{S}'$ to the learner. The goal of the learner is to output a concept $\hat{f}: \Rn \rightarrow \{-1, 1\}$, such that with probability $1-\delta$ (over the randomness of the samples and all internal random bits of the learning algorithm), $\Pr_{x \sim D}( \hat{f}(x) \neq f^*(x) ) \leq \epsilon$ for any prescribed error rate $\epsilon \in (0, 1)$ and failure probability $\delta \in (0, 1)$. We say an algorithm PAC learns the class $\calH_{d, K}$ if the guarantee holds uniformly for any member $f^* \in \calH_{d, K}$.
\end{definition}

%It is worth mentioning that when corrupting the clean sample set $\bar{S}$, the adversary has full knowledge of the distribution $D$, the concept $f^*$, the learning algorithm including its internal random bits, and has unbounded computational power.

\citet{bshouty2002pac} presented a computationally {\em inefficient} algorithm for learning any concept class with near-optimal sample complexity and noise tolerance as far as the concept class has finite VC-dimension (see Theorem~7 therein). Since the VC-dimension of $\calH_{d, K}$ is $O(K^d \log n)$, we have:
\begin{fact}
There exists an inefficient algorithm that PAC learns $\calH_{d, K}$ with near-optimal sample complexity $O(K^d \log n)$ and noise tolerance $\Omega(\epsilon)$.
\end{fact}
Designing efficient algorithms that match such statistical guarantees thus becomes a core research line.

On the one hand, for the special case of homogeneous sparse halfspaces,  \citet{shen2021attribute} gave a state-of-the-art algorithm with sample complexity $O(K^2 \log^5 n)$ and noise tolerance $\Omega(\epsilon)$ when the instance distribution $D$ is isotropic log-concave. On the other hand, for learning general sparse low-degree PTFs, very little is known, since the structure of PTFs is tremendously complex. To our knowledge, it appears that the only known approach is to reducing the problem to a generic approach proposed in the  early work of \citet{kearns1988learning}. In particular, Theorem~12 therein implies that any concept class $\calH$ can be PAC learned with nasty noise in polynomial time provided that there exists a polynomial-time algorithm that PAC learns it in the absence of noise and that $\eta \leq O\big( \frac{\epsilon}{\vcdim(\calH)} \log \frac{\vcdim(\calH)}{\epsilon} \big)$, where $\vcdim(\calH)$ denotes the VC-dimension of $\calH$. We therefore have the following (see Appendix~\ref{sec:app:prem} for the proof):

%As the noise-free problem of learning low-degree PTFs can be solved efficiently via linear programming \cite{maass1994fast}, 

\begin{fact}\label{fact:baseline}
There exists an efficient algorithm that draws $C_0 \cdot \frac{K^{3d} \log^3 n}{\epsilon^6} \log\frac{1}{\delta}$ samples from the nasty adversary and PAC learns $\calH_{d, K}$ provided that $\eta \leq O\big(\frac{d\epsilon}{K^d}\log\frac{1}{\epsilon}\big)$, where $C_0 > 0$ is an absolute constant.
\end{fact}

The result above is appealing since it makes no distributional assumption and it runs in polynomial time. However, the main issue is on the noise tolerance: the robustness of the algorithm degrades significantly when $K$ is large. For example, in the interesting regime $K = \Theta(\log n)$, the noise tolerance is dimension-dependent, meaning that the algorithmic guarantees are brittle in high-dimensional problems. See  \citet{kalai2005agnostic,klivans2009learning,long2011learning,awasthi2017power,shen2021sample,shen2021attribute} and a comprehensive survey by \citet{diakonikolas2019recent} for the importance and challenges of obtaining dimension-independent noise tolerance.

%In this regard, a large volume of efficient algorithms have been designed to obtain {\em (VC)-dimension-independent} noise tolerance, though for problems different from ours \cite{kalai2005agnostic,klivans2009learning,long2011learning,awasthi2017power,shen2021sample,shen2021attribute}; see also a comprehensive survey by \citet{diakonikolas2019recent} which describes the merit in unsupervised learning such as robust mean  estimation. We note that \citet{diakonikolas2018learning} gave the first such algorithm for learning low-degree PTFs, yet without attribute efficiency.

\subsection{Main results}

Throughout the paper, we always assume:
\begin{assumption}\label{as:D}
$D$ is the  Gaussian distribution $\calN(0,I_{n\times n})$. 
\end{assumption}

Our main result is an attribute-efficient algorithm that runs in time $({nd}/{\epsilon})^{O(d)}$ and PAC learns $\calH_{d, K}$ with {\em dimension-independent} noise tolerance.

\begin{theorem}[Theorem~\ref{thm:main}, informal]\label{thm:main-informal}
Assume that $D$ is the standard Gaussian distribution $\calN(0, I_{n\times n})$. There is an algorithm that runs in time $({nd}/{\epsilon})^{O(d)}$ and PAC learns  $\calH_{d, K}$ by drawing $C \cdot \frac{ K^{4d} (d\log n)^{5d}}{\epsilon^{2d+2}}$ samples from the nasty adversary for some absolute constant $C>0$, provided that $\eta \leq O(\epsilon^{d+1}/d^{2d})$. 
\end{theorem}

\begin{remark}[Sample complexity]\label{rmk:sample-size}
It is known that for {\em efficient and outlier-robust} algorithms, $\Omega(K^2)$ samples are necessary to obtain an error bound of $O(\epsilon)$ even for   linear models \cite{diakonikolas2017statistical}. Thus, the multiplicative factor $K^{4d}$ in our sample complexity bound is very close to the best possible scaling of $K^{2d}$ and the best known result in Fact~\ref{fact:baseline}. The exponent $d$ in the  factor $\frac{1}{\epsilon^{2d+2}}$ comes from our two-step approach: we will first robustly estimate the Chow vector \cite{chow61} of $f^*$ up to error $\epsilon_0$ using $\Omega({1}/{\epsilon_0^2})$ samples, and then apply an algorithmic result of \citet{TTV09,DDFS14,diakonikolas2018learning} to construct a PTF with misclassification error rate of $O(d \cdot \epsilon_0^{1/{d+1}})$. This in turn suggests that we have to set $\epsilon_0 = (\epsilon/d)^{d+1}$ in order to get the target error rate $\epsilon$. As noted in \citet{diakonikolas2018learning}, such overhead on the scaling of $\epsilon$ is inherent when using only Chow vector to establish PAC guarantees for degree-$d$ PTFs.%; see also \citet{EI22}.
\end{remark}

\begin{remark}[Noise tolerance]
Our noise tolerance matches the best known one given by \citet{diakonikolas2018learning}, which studied non-sparse low-degree PTFs. When the degree $d$ is a constant, the noise tolerance reads as $\epsilon^{\Omega(1)}$, qualitatively matching the information-theoretic limit of $\Omega(\epsilon)$. Yet, the existence of efficient low-degree PTF learners with optimal noise tolerance is widely open.
\end{remark}

\begin{remark}[Comparison to prior works]
Most in line with this work are \citet{shen2021attribute} and \citet{diakonikolas2018learning}. \citet{shen2021attribute} gave a state-of-the-art algorithm for learning $K$-sparse halfspaces, but their algorithm cannot be generalized to learn sparse low-degree PTFs. \citet{diakonikolas2018learning} developed an efficient algorithm for learning non-sparse PTFs. Their sample complexity bound reads as $\frac{1}{\epsilon^{2d+2}} \cdot (nd)^{O(d)}$, which is not attribute-efficient and thus is inapplicable for real-world problems where the number of samples is orders of magnitude less than that of attributes. At a high level, our result can be thought of as a significant generalization of both. 
\end{remark}

\begin{remark}[Running time]
The computational cost of our algorithm is $(nd/\epsilon)^{O(d)}$ which we believe may not be significantly improved in the presence of the nasty noise. This is because the adversary has the power to inspect the algorithm and to corrupt any samples, which forces any robust algorithm to carefully verify the covariance matrix whose size is $n^{O(d)} \times n^{O(d)}$; see Section~\ref{subsec:techniques} and Section~\ref{sec:alg}. Under quite different problem settings, prior works leveraged the underlying sparsity for improved computational complexity; see Section~\ref{subsec:related}.
%We consider an algorithm to be computationally efficient when its running time depend polynomially on $(\frac{n}{\epsilon})^d$. Unlike the goal in attribute-efficient learrning that requries the sample complexity to depend polylogarithmically on $n$, the computational cost of our algorithm is $n^{O(d)}$. This is because under the nasty noise, if without any additional information, the algorithm needs to scan at least one time through all $n^d$ basis polynomials for a degree-$d$ $n$-variate PTF to confirm its statistical property, which is unavoidable in any spectral technique that requires calculating the empirical covariance matrix of size $n^d\times n^d$ (See Section~\ref{sec:alg}). Though orthogonal to this work, it is widely open whether we can incorporate the existing computational techniques to achieve an attribute-efficient running time.
\end{remark}

\subsection{Overview of main techniques}\label{subsec:techniques}

Our starting point to learn sparse low-degree PTFs is an elegant algorithmic result from \citet{diakonikolas2018learning}, which shows that as far as one is able to approximate the Chow vector  of $f^*$ \cite{chow61}, it is possible to reconstruct a PTF $\hat{f}$ with  PAC guarantee in time $({nd}/{\epsilon})^{O(d)}$. To apply such scheme, we will need to 1)~properly define the Chow vector since it depends on the choice of the basis of polynomials; and 2)~estimate the Chow vector of $f^*$ in time $({nd}/{\epsilon})^{O(d)}$. Our technical novelty  lies in new structural and algorithmic ingredients to achieve  attribute efficiency.

\medskip
\noindent
{\bfseries 1) Structural result: attribute sparsity induces sparse Chow vectors under the basis of Hermite polynomials.}
Prior works such as the closely related work of \citet{diakonikolas2018learning} tend to use the basis of monomials to define Chow vectors. However, there is no guarantee that such definition would exhibit the desired sparsity structure. For example, consider that $D$ is the standard Gaussian distribution. For $K$-sparse degree-$d$ PTFs on $\Rn$, the number of monomials with non-zero coefficients can be as large as $ (\frac{n}{d})^{\lfloor\frac{d}{2}\rfloor}$ for any $K \leq n/2$, $d\geq2$ (see Lemma~\ref{lem:K-k-monomial}).  
%$(n+1)^d - (n-K+1)^d \geq (n/2)^{d-1}$ 
%Thus, the Chow vectors with respect to monomials do not exhibit appropriate sparsity structure.
Our first technical insight is that, the condition that a degree-$d$ PTF is $K$-sparse implies a $k$-sparse Chow vector with respect to the basis of Hermite polynomials of degree at most $d$ ($k$ is roughly $2K^d$, see Lemma~\ref{lem:K-k}). 
This endows a sparsity structure of the Chow vector of $f^*$, which in turn is leveraged into our algorithmic design since we can now focus on a much narrower space of Chow vectors and thus lower sample complexity. 

{It is worth mentioning that while we present our results under Gaussian distribution and thus use the Hermite polynomials as the basis to ease analysis, the choice of basis can go well beyond that; see Appendix~\ref{sec:app:basis} for more discussions.}

\medskip
\noindent
{\bfseries 2) Algorithmic result: attribute-efficient robust Chow vector estimation.} Denote by $m(x)$ the vector of all $n$-variate Hermite polynomials of degree at most $d$. The Chow vector (also known as Fourier coefficients) of a Boolean-valued function $f: \Rn \rightarrow \{-1, 1\}$ is defined as $\Chow_f := \EXP_{x \sim D}[ f(x) \cdot m(x)]$.
As we discussed, $\Chow_{f^*}$ is $k$-sparse on an unknown support set. To estimate it within error $\epsilon_0$ in $\ell_2$-norm,  it suffices to find some $k$-sparse vector $u$, such that for {\em all} $2k$-sparse unit vector $v$, we have $\abs{\inner{v}{u - \Chow_{f^*}}} \leq \epsilon_0$ (see Lemma~\ref{lem:inner-l2}). We will choose $u$ as  an empirical Chow vector, i.e. $u = \sum_{(x, y) \in \bar{S}''} y \cdot m(x)$, where $\bar{S}''$ needs to be a carefully selected subset of $\bar{S}'$. Now recent developments in noise-tolerant classification \cite{awasthi2017power,shen2021attribute,shen2023linearmalicious} suggest that such estimation error is governed by the maximum eigenvalue  on all possible $2k$-sparse directions of the empirical covariance matrix\footnote{We will slightly abuse the terminology of covariance matrix to refer to the one without subtracting the mean.} $\Sigma := \frac{1}{\abs{\bar{S}''}} \sum_{x \in \bar{S}''} m(x) m(x)\trans$. This structural result can be cast into algorithmic design: find a large enough subset $\bar{S}''$ such that the maximum sparse eigenvalue of $(\Sigma - I)$ is close to zero (note that $\EXP_{x\sim D}[m(x) m(x)\trans] = I$). Unfortunately, there are two technical challenges: first, computing the maximum sparse eigenvalue is NP-hard; second, searching for such a subset is also computationally intractable.

\medskip
\noindent
{\bfseries 2a) Small Frobenius norm certifies a good approximation.} We tackle the first challenge by considering a sufficient condition: if the Frobenius norm of $(\Sigma-I)$ restricted on its $(2k)^2$ largest entries (in magnitude) is small, then so is the maximum sparse eigenvalue~--~this would imply the empirical Chow vector be a good approximation to $\Chow_{f^*}$; see Theorem~\ref{thm:small-fro}.

\medskip
\noindent
{\bfseries 2b) Large Frobenius norm validates a filter.} It remains to show that when the restricted Frobenius norm is large, say larger than some carefully chosen parameter $\kappa$, how to find a proper subset $\bar{S}''$ that is clean enough, in the sense that its  distributional properties act almost as it be an uncorrupted sample set. Our key idea is to construct a polynomial $p_2$ such that (i) its empirical mean on the current sample set $\bar{S}'$ equals the restricted Frobenius norm (which is large); and (ii) it has small value on clean samples. These two properties ensure that there must be a noticeable fraction of samples in $\bar{S}'$ that caused a large function value of $p_2$, and they are very likely the corrupted ones; these will then be filtered to produce the new sample set $\bar{S}''$ (Theorem~\ref{thm:large-fro}). In addition, the polynomial $p_2$ is constructed in such a way that it is {\em sparse} under the basis $\{ m_i(x)\cdot m_j(x)\}$, for the sake of attribute efficiency. %We then show that with a proper setting on the threshold $\kappa$, there must be many (corrupted) samples that violate such tail bound, 

%\medskip

We check the Frobenius norm condition every time a new sample set $\bar{S}''$ is produced, and show that after a finite number of phases, we must be able to obtain a clean enough sample set $\bar{S}''$ that allows us to output a good estimate of the Chow vector $\Chow_{f^*}$ (Theorem~\ref{thm:main-chow}).

We note that the idea of using Frobenius norm as a surrogate of the maximum sparse eigenvalue value has been explored in \citet{diak2019sparsemean,zeng2022sparsemean} for robust sparse mean estimation. In those works, the Frobenius-norm condition was combined with a localized eigenvalue condition to establish their main results, while we discover that the Frobenius norm itself suffices for our purpose. This appears an interesting and practical finding as it reduces the computational cost and simplifies algorithmic design.

%\subsection{Suboptimality of reduction approaches}
%
%In view of our structural result (i.e. the first part in Section~\ref{subsec:techniques}), we can design two natural reduction-based approaches to learning $K$-sparse degree-$d$ PTFs. The first is to reduce to learning of sparse halfspaces, and the second is to reduce to sparse mean estimation.
%
%In the first approach, we can think of any degree-$d$ polynomial $p(x) = {v} \cdot {m(x)}$, where the number of non-zero elements in $v$ is less than $k$. Thus, any $K$-sparse PTF $f(x) = \sign( v \cdot m(x))$ can be thought of as a $k$-sparse halfspace on the space expanded by $m(x)$. The catch is that even though an uncorrupted $x$ is drawn from $D$, the Gaussian distribution, the expanded instance $m(x)$ is not; it has much heavier tail (see Fact~\ref{fact:tail-bound}). This prevents the application of most prior noise-tolerant algorithms, which learn the class of halfspaces only under strong distributional assumptions \cite{kalai2005agnostic,klivans2009learning,awasthi2017power,shen2021attribute}. Surprisingly, it turns out that the very early work of \cite{kearns1988learning} is applicable here since it is distribution-independent. However, their result only leads to a dimension-dependent noise tolerance of $\Omega(\epsilon/n^d)$ while ours is dimension-independent. This appears a major shortcoming of their algorithm; see \cite{kalai2005agnostic,klivans2009learning,awasthi2017power,shen2021attribute} for how researchers made tremendous amount of efforts in order to obtain dimension-independent bound.

\subsection{Related works}\label{subsec:related}

The problem of learning from few samples dates back to the 1980s, when practitioners were confronting a pressing challenge: the number of samples available is orders of magnitude less than that of attributes, making classical algorithms fail to provide  guarantees. The challenge persists even in the big data era, since in many domains such as healthcare, there is a limited availability of samples (i.e. patients) \cite{candes2008introduction}. This has motivated a flurry of research on attribute-efficient learning of sparse concepts. A partial list of interesting works includes \cite{littlestone1987learning,blum1990learning,chen1998atomic,tibshirani1996regression,tropp2004greed,candes2005decoding,foucart2011hard,plan2013robust,shen2017iteration,shen2017partial,shen2018tight,wang2018provable,shen2020one} that studied linear models in the absence of noise (or with benign noise). Later, \citet{candes2013phaselift,netrapalli2013phase,candes2015phase} studied the problem of phase retrieval which can be seen as learning sparse quadratic polynomials. The setup was  generalized in \citet{chen2020sparsepoly} which studied efficient learning of sparse low-degree polynomials. The work of \cite{zeng2022cs,zeng2023semi} implied efficient PAC learning of sparse PTFs when only labels are corrupted under a crowdsourcing PAC model \cite{awasthi2017efficient}.

In the presence of the nasty noise, the problem becomes subtle. Without distributional assumptions on $D$, it is known that even for the special case of learning  halfspaces under the adversarial label noise, it is computationally hard when the noise rate is $\epsilon$ \cite{guruswami2006hardness,feldman2006new,daniely2016complexity}. Thus, distribution-independent algorithms are either unable to tolerate the nasty noise at a rate greater than $\epsilon$ \cite{kearns1988learning}, or runs in super-polynomial time \cite{bshouty2002pac}. This motivates the study of efficient algorithms under distributional assumptions \cite{kalai2005agnostic,klivans2009learning,awasthi2017power,shen2021attribute,shen2023linearmalicious}, which is the research line we follow. In unsupervised learning such as mean and covariance estimation, similar noise models are investigated broadly in recent years since the seminal works of \citet{diakonikolas2016robust,lai2016agnostic}; see \citet{diakonikolas2019recent} for a comprehensive survey.

The interplay between sparsity and robustness is more involved than it appears to. Under the statistical-query framework, \citet{diakonikolas2017statistical} showed that any efficient and robust algorithms must draw $\Omega(K^2)$ samples in the presence of the nasty noise, complementing sample complexity upper bounds obtained in recent years \cite{balakrishnan2017computation,diak2019sparsemean,shen2021attribute,dia2022sparseME}. This is in stark contrast to learning with label noise, where $O(K)$ sample complexity can be established \cite{zhang2018efficient,zhang2020efficient,shen2021power}.

We note that orthogonal to exploring sparsity for improved sample complexity, there are elegant works that explore sparsity for improved computational complexity for learning Boolean-valued functions \cite{HS07sparsePTF,andoni2014sparsepoly}, or using low-degree PTFs as primitives to approximate other concepts such as halfspaces \cite{kalai2005agnostic} and decision lists \cite{servedio2012attribute}.

Lastly, it is worth mentioning that low-rank matrix estimation \cite{candes2009exact,candes2011robust,xu2012robust,xu2013outlier,shen2014online,shen2016online,shen2016learning}, or more specifically, the one-bit variant \cite{davenport2014bit}, is a relevant field but little progress has been made towards algorithmic robustness \cite{shen2019robust}.

\subsection{Roadmap}

We collect  notations and definitions in Section~\ref{sec:setup}. The main algorithms are described in Section~\ref{sec:alg} with a few lemmas to illustrate the idea, and the primary performance guarantees are stated in Section~\ref{sec:guarantee}. We conclude this work in Section~\ref{sec:con}. All omitted proofs can be found in the appendix.

%% file: setup.tex
\section{Preliminaries}\label{sec:setup}

{\bfseries Vectors and matrices.} \
For a vector $v \in \Rn$, we use $v_i$ to denote its $i$-th element. For two vectors $u$ and $v$, we write $u \cdot v$ as the inner product in the Euclidean space. We denote by $\onenorm{v}$, $\twonorm{v}$, $\infnorm{v}$ the $\ell_1$-norm, $\ell_2$-norm, and $\ell_{\infty}$-norm of $v$ respectively. The support set of a vector $v$ is the index set of its non-zero elements, and $\zeronorm{v}$ denotes the cardinality of the support set. We will use the hard thresholding operator $\hardthr_k(v)$ to produce a $k$-sparse vector: the $k$ largest (in magnitude) elements of $v$ are retained and the rest are set to zero. 
Let $\Lambda \subset [n]$ where $[n] := \{1, \dots, n\}$. The restriction of $v$ on $\Lambda$, $v_{\Lambda}$, is obtained by keeping the elements in $\Lambda$ while setting the rest to zero. 

Let $A$ and $B$ be two matrices in $\R^{n_1\times n_2}$. We write $\tr(A)$ as the trace of $A$ when it is square, and write $\inner{A}{B} := \tr(A\trans B)$. We denote by $\fronorm{A}$ the Frobenius norm, which equals $\sqrt{\inner{A}{A}}$. We will also use $\zeronorm{A}$ to count the number of non-zero entries in $A$. Let $U \subset [n_1] \times [n_2]$. The restriction of $A$ on $U$, $A_U$, is obtained by keeping the elements in $U$ but setting the rest to zero.

\medskip
\noindent
{\bfseries Probability, $L^2$-space.} \
Let $D$ be a distribution on $\Rn$ and  $p$ be a function with the same support of $D$. We denote by $\EXP_{X\sim D}[p(X)]$ the expectation of $p$ on $D$. Let $S$ be a finite set of instances. We write $\EXP_{X\sim S}[p(X)] := \frac{1}{\abs{S}} \sum_{x \in S} p(x)$ as the empirical mean of $p$ on $S$. To ease notation, we will often use $\EXP[p(D)]$ in place of $\EXP_{X\sim D}[p(X)]$, and likewise for $\EXP[p(S)]$. Similarly, we will write $\Pr( p(D) > t) := \Pr_{X\sim D}( p(X) > t)$, and  $\Pr( p(S) > t ) := \Pr_{X\sim S}( p(X) > t)$ where $X \sim S$ signifies uniform distribution on $S$. %Let $E$ be some event. We 

The $L^2(\Rn, D)$ space is equipped with the inner product $\inner{p}{q}_D := \EXP_{x \sim D}[ p(x) \cdot q(x)]$ for any functions $p$ and $q$ on $\Rn$.
The induced $L^2$-norm of a function $p$ is given by $\norm{p}_{L^2(D)} := \sqrt{\inner{p}{p}_D} = \sqrt{\EXP[p^2(D)]}$, which we will simply write as $\Ltwonorm{p}$  when $D$ is clear from the context.

\medskip
\noindent
{\bfseries Polynomials.} \
Denote by $\polyclass$ the class of polynomials on $\Rn$ with degree at most $d$. A degree-$d$ {\em polynomial threshold function} (PTF) is of the form $f(x) = \sign( p(x))$ for some $p \in \polyclass$. 
Denote by $\hermite_d(x) = \frac{1}{\sqrt{d!}} (-1)^d \cdot e^{-\frac{x^2}{2}} \frac{\mathrm{d}^d}{\mathrm{d} x^d}e^{-\frac{x^2}{2}}$ the normalized univariate degree-$d$ {\em Hermite polynomial} on $\R$. The normalized $n$-variate Hermite polynomial is given by $\hermite_{\ba}(x) = \prod_{i=1}^{n}\hermite_{\ba_i}(x_i)$ for some multi-index $\ba \in \mathbb{N}^n$; for brevity we  refer to them as Hermite polynomials. It is known that $\hermite_{\leq d} := \{ \hermite_{\ba}: \ba \in \mathbb{N}^n, \onenorm{\ba} \leq d \}$ form a complete orthonormal basis for polynomials of degree at most $d$ in $L^2(\Rn, D)$; see Prop.~11.33 of \citet{donnell2014book}. It is easy to see that $\hermite_{\leq d}$ contains $M := 1 + \sum_{t=1}^d \binom{t + n - 1}{t}$ members; we collect them as a vector $m(x) = \big(m_1(x), \dots, m_M(x)\big)$, with the first element $m_1(x) \equiv 1$. In our analysis, it suffices to keep in mind that $M < (n+1)^d$.

%Denote $M = (n+1)^d$. It is known that there exists a complete orthonormal basis $m(x) = \{ m_i(x) \}_{i=1}^M$ with $m_i \in \polyclass$, such that any $p \in \polyclass$ is a linear combination of elements in $m(x)$ in the $L^2$ space. In this paper, we always consider $D$ as the standard Gaussian distribution $N(0, I_{n\times n})$, and thus, $m(x)$ consists of all Hermite polynomials with degree at most $d$. 

Given the vector $m(x)$ and the distribution $D$, the {\em Chow vector} \cite{chow61,diakonikolas2018learning} of a Boolean-valued function $f: \Rn \rightarrow \{-1, 1\}$ is defined as follows:
\vspace{-0.1in}
\begin{equation}
\Chow_f := \EXP_{x \sim D}[ f(x) \cdot m(x)],
\end{equation}
where we multiplied each element of $m(x)$ by $f(x)$.

\begin{definition}[Sparse polynomials and PTFs]
We say a polynomial $p \in \polyclass$ is  $K$-sparse if there exists an index set $\Lambda \subset [n]$ with $\abs{\Lambda} \leq K$, such that $p(x) =q(x_{\Lambda})$ for some  $q \in \polyclass$. We say a PTF $f(x) = \sign( p(x))$ is $K$-sparse if $p$ is $K$-sparse. The class of $K$-sparse PTFs on $\Rn$ with degree at most $d$ is denoted by $\calH_{d, K}$.
\end{definition}

One important observation is that our definition of sparse polynomials implies a sparsity pattern in the Chow vector; see Appendix~\ref{sec:app:prem} for the proof.

\begin{lemma}\label{lem:K-k}
Let $f$ be a $K$-sparse degree-$d$ PTF. Then $\Chow_f$ is a $k$-sparse vector under the basis of Hermite polynomials, where $k = d+1$ if $K=1$ and $k \leq 2K^d$ otherwise.
\end{lemma}

As we discussed in Section~\ref{subsec:techniques}, there will be two concept classes involved in our algorithm and analysis. The first is the class of polynomials that have a sparse Chow vector under the basis of $m(x)$:
\begin{align}\label{eq:P1}
\sparselinearclass := \big\{ p_1: x \mapsto \inner{v}{{m}(x)}, v \in \R^{(n+1)^d},  \twonorm{v} = 1, \zeronorm{v} \leq 2k \big\},
\end{align}
which will be useful in characterizing the approximation error to the Chow vector of interest. Another class consists of quadratic terms in $m(x)$,
\begin{align}\label{eq:P2}
\sparseclass := \big\{ &p_2: x \mapsto  \langle A_U, m(x) m(x)\trans - I \rangle,  U\trans = U,  \zeronorm{U} \leq s, A \in \mathbb{S}^{(n+1)^d}, \fronorm{A_U} = 1 \big\}
\end{align}
where $\mathbb{S}^{(n+1)^d} := \{ A: A \in \R^{(n+1)^d \times (n+1)^d}, A\trans = A\}$.
Note that the polynomials in $\sparseclass$ have degree at most $2d$, and can be represented as a linear combination of at most $s$ elements of the form $m_i(x) m_j(x)$. They will be used to construct certain distributional statistics based on the empirical samples for filtering. 

We will often use subscript to stress the membership of a polynomial in either class: we will write $p_1 \in \sparselinearclass$ and $p_2 \in \sparseclass$, rather than using the subscript for counting.

\medskip
\noindent
{\bfseries Reserved symbols.} Throughout the paper, $K$ always refers to the number of non-zero attributes that a sparse PTF depends on, and $k$ is the sparsity of the Chow vector under the Hermite polynomials $m(x)$ (see Lemma~\ref{lem:K-k}). We reserve $\epsilon, \delta, \eta$ as in Definition~\ref{def:nasty}. We wrote $\bar{S}'$ as the corrupted sample set, and $S'$ as the one without labels.

The capital and lowercase letters $C$ and $c$, and their subscript variants, are always used to denote some absolute constants, though we do not track closely their values. We reserve 
\begin{equation}\label{eq:params}
\gamma = \big(C_1 d \cdot \log\frac{nd}{\epsilon\delta}\big)^{d/2},\  \rho_2 = C_2 \cdot d^{\frac34} \cdot (c_0 d)^d.
\end{equation}
As will be clear in our analysis, $\gamma$ upper bounds $\max_{x\in S }\infnorm{m(x)}$ for $S$ drawn from $D$. Thus, we will only keep samples in $\bar{S}'$ with $x \in \Xgamma$ where 
\begin{equation}
\Xgamma := \{x \in \Rn: \infnorm{m(x)} \leq \gamma\}.
\end{equation}
The quantity $\rho_2$ upper bound $\Ltwonorm{p_2}$; see  Lemma~\ref{lem:p2-properties}.

%% file: alg.tex
\begin{algorithm}[t]
\caption{Main Algorithm: Attribute-Efficient Robust Chow Vector Estimator}
\label{alg:main}
\begin{algorithmic}[1]

\REQUIRE A nasty adversary $\oraclexy(\eta)$ with $\eta \in [0, \frac12 - c]$ for some absolute constant $c \in (0, \frac12]$, hypothesis class $\calH_{d, K}$ that contains $f^*$, target error rate $\epsilon \in (0, 1)$, failure probability $\delta \in (0, 1)$.

\ENSURE A sparse vector $u \in \R^{(n+1)^d}$.

\STATE $\bar{S}' \gets$ draw  $C \cdot \frac{d^{5d} K^{4d}}{\epsilon^2} \log^{5d}\big( \frac{nd}{\epsilon\delta} \big)$ samples from $\oraclexy(\eta)$.\label{step:main-sample}

\STATE $k \gets d+1$ if $K=1$ or $k\gets 2K^d$ if $K > 1$.\label{step:main-k}

%\STATE $\gamma \gets \big(C_1 d \cdot \log\frac{nd}{\epsilon\delta}\big)^{d/2}$, $\rho_2 \gets {C_2} \cdot d^{\frac34} \cdot \big( c_0 d \big)^{d}$.

\STATE $\kappa \gets \frac{28}{c^2}  \cdot \big[  \rho_2  \cdot ( c_0 \log\frac{1}{\eta} + c_0  d)^d \cdot \eta + \epsilon \big]$.

\STATE $l_{\max} \gets \frac{4\eta k \gamma^2}{\epsilon}+1$.

\STATE $\bar{S'_1} \gets \bar{S'} \cap \{(x, y): \infnorm{m(x)} \leq \gamma \}$.\label{step:main-prune}

%\Jie{maybe we will need infinity norm}

%$\forall x\in S'$, {\bf if} $\twonorm{m(x)}\geq T_{\max}^2/2$, {\bf then} remove $x$ from $S'$. %\Jie{use $\twonorm{m(x)}$ will help remind readers that $m(x)$ is a vector.}

\FOR{phase $l = 1$ to $l_{\max}$}

\STATE $\Sigma \gets  \EXP_{x\sim S'_l}[m(x)m(x)\trans]$, $\{(i_t, j_t)\}_{t= 1}^{4k^2} \gets$ index set of the largest (in magnitude) $2k$ diagonal entries and $2k^2 - k$ entries above the main diagonal of ${\Sigma} - I$. $U \gets \{(i_t, j_t)\}_{t\geq 1} \cup \{(j_t, i_t)\}_{t\geq 1}$.\label{step:main-Sigma}
%{\color{red} $J \gets \{i_t\}_{t \geq 1} \cup \{j_t\}_{t\geq 1}$, $U' \gets J \times J$.}
%$J \gets \{i_t\}_{t \geq 1} \cup \{j_t\}_{t\geq 1}$, $U \gets J \times J$, $U' \gets \{(i_t, j_t)\}_{t\geq 1} \cup \{(j_t, i_t)\}_{t\geq 1}$.

\STATE {\bfseries if} $\fronorm{(\Sigma-I)_{U}} \leq \kappa$ {\bfseries then return} $u \gets \hardthr_k\big(\EXP_{(x, y)\sim \bar{S}'_l}[y\cdot m(x)]\big)$.\label{step:main-return}

\STATE $S'_{l+1} \gets \filter(S'_l, U, \Sigma, k, \gamma, \rho_2)$.\label{step:main-filter}

\ENDFOR

\end{algorithmic}
\end{algorithm}

\begin{algorithm}[t!]
\caption{$\filter(S', U, \Sigma, k, \gamma, \rho_2)$}\label{alg:filter}

\begin{algorithmic}[1]

\STATE $A \gets \frac{1}{\fronorm{(\Sigma-I)_U}} (\Sigma - I)_U$.

\STATE $p_2(x) \gets \langle A, { m(x) m(x)\trans - I} \rangle$.  

\STATE Find $t \in (0, 4k \gamma^2)$ such that \label{step:filter-quad}

\vspace{-0.25in}
\begin{equation*}
\Pr\big( \abs{p_2(S')}\geq t \big) \geq 6 \exp\big(-(t/\rho_2)^{1/d} / c_0  \big) + \frac{3\epsilon}{k\gamma^2}.
\end{equation*}

\vspace{-0.1in}

\STATE {\bfseries return} $S'' \leftarrow \{x\in S': \abs{p_2(x)}\leq t\}$.

\end{algorithmic}
\end{algorithm}

\section{Main Algorithms}\label{sec:alg}

Our main algorithm, Algorithm~\ref{alg:main}, aims to approximate the Chow vector of the underlying polynomial threshold function $f^* \in \calH_{d, K}$ by drawing a small number of samples from the nasty adversary. Observe that the setting of $k$ at Step~2 follows from Lemma~\ref{lem:K-k}, i.e. $k$ is the sparsity of $\Chow_{f^*}$ under the basis of Hermite polynomials. With this in mind, we design three sparsity-induced components in the algorithm: pruning samples that must be outliers (Step~5), certifying that the sample set is clean enough and returning the empirical Chow vector (Step~8), or filtering samples with a carefully designed condition (Step~9). We elaborate on each component in the following.

\subsection{Pruning}

Since the outliers created by the nasty adversary may be arbitrary, it is useful to design some simple screening rule to remove samples that must have been corrupted. In this step, we leverage the distributional assumption that $D$, the distribution of instances, is the standard Gaussian $\calN(0, I_{n\times n})$. Since the concept class consists of polynomials with degree at most $d$, it is known that any Hermite polynomial $m_i(x)$ must concentrate around its mean with a known tail bound \cite{janson1997gaussian}. As  the mean of $m_i(x)$ equals zero for all $i \neq 1$ (recall that $m_1(x) \equiv 1$), it is possible to specify a certain radius $\gamma$ for pruning. Similar to \citet{zeng2022sparsemean}, we apply the $\ell_{\infty}$-norm metric for attribute efficiency, that is, we remove all samples $(x, y)$ in $\bar{S}'$ satisfying $\infnorm{m(x)} > \gamma$. The following lemma shows that with high probability, no clean sample will be pruned under a proper choice of $\gamma$.

\begin{lemma}\label{lem:gamma}
Let $S$ be a set of samples drawn independently from $D$.
With probability at least $1- \delta_{\gamma}$, we have $\max_{x \in S} \infnorm{m(x)} \leq \gamma $ where $\gamma :=  \big( c_0 \log\frac{\abs{S}(n+1)^d}{\delta_{\gamma}}  \big)^{d/2}$.
\end{lemma}

Recall that the concrete value of $\gamma$ is given in \eqref{eq:params}; it  is obtained by setting $\abs{S}$ as the same size as in Step~1 of Algorithm~\ref{alg:main} and setting $\delta_{\gamma} = \frac{\epsilon^2 \delta}{64 \rho_2^2}$ (note that $\delta_{\gamma} \leq O(\delta)$). The appearance of $\delta$ in $\delta_{\gamma}$ is not surprising. For the multiplicative factor $\frac{\epsilon^2}{64\rho_2^2}$, technically speaking, it ensures that the total variation distance between the distribution $D$ conditioned on the event $x \in \Xgamma$ and $D$ is $O(\epsilon)$, thus one can in principle consider uniform concentration on the former to ease analysis (since it is defined on a bounded domain); see Proposition~\ref{prop:sample-comp}.

\subsection{Filtering}\label{subsec:alg:filter}

At Step~7 of Algorithm~\ref{alg:main}, we compute the empirical covariance matrix $\Sigma$ and the index set $U$ of the $(2k)^2$ largest entries (in magnitude) of $\Sigma - I$. As we highlighted in Section~\ref{subsec:techniques}, this is a computationally efficient way to obtain an upper bound on the maximum eigenvalue of $\Sigma-I$ on all  $2k$-sparse directions. The structural constraint on $U$ comes from the observation that for $2k$-sparse $v$, we have $v\trans (\Sigma-I) v = \inner{\Sigma-I}{vv\trans}$ and $vv\trans$ has $2k$ non-zero diagonal entries and $4k^2 - 2k$ off-diagonal symmetric entries.

If the restricted Frobenius norm, $\fronorm{(\Sigma-I)_U}$, is greater than some threshold $\kappa$, Algorithm~\ref{alg:main} will invoke a filtering subroutine, Algorithm~\ref{alg:filter}, to remove samples that were potentially corrupted. The high-level idea of Algorithm~\ref{alg:filter} follows from prior works on robust mean estimation \cite{diakonikolas2016robust,diak2019sparsemean,zeng2022sparsemean}: under the condition that a certain measure of the empirical covariance matrix is large, there must be some samples that behave in quite a different way from those drawing from $D$. Our technical novelty is a new algorithm and analysis showing that the Frobenius norm itself suffices to validate a certain type of test that can identify those potentially corrupted samples~--~this is a new feature as existing robust sparse mean estimation algorithms \cite{diak2019sparsemean,zeng2022sparsemean} rely on a combination of the Frobenius norm and a localized eigenvalue condition. An immediate implication of our finding is that one can expect lower computational cost of our algorithm due to the lack of eigenvalue computation.

Now we discuss how to design a test to filter potentially corrupted samples. The idea is to create a sample-dependent polynomial $p_2$ with the following two properties: 1) its empirical mean on $S'$ equals $\fronorm{(\Sigma-I)_U}$; and 2) $p_2$ is small (in expectation) on uncorrupted samples. In this way, since we have the condition that $\fronorm{(\Sigma-I)_U}$ is large, there must be a noticeable fraction of samples in $S'$ that correspond to large function values of $p_2$. This combined with the second property suffice to identify  abnormal samples.

Indeed, since $\Sigma = \EXP_{x \sim S'}[m(x) m(x)\trans]$, we can show that
\begin{align*}
\fronorm{(\Sigma-I)_U} =&\ \frac{1}{\fronorm{(\Sigma-I)_U}} \inner{{(\Sigma-I)_U}}{\EXP_{x \sim S'}[m(x) m(x)\trans] - I}\\
=&\ \EXP_{x \sim S'}\Big[ \frac{1}{\fronorm{(\Sigma-I)_U}} \inner{{(\Sigma-I)_U}}{m(x) m(x)\trans - I} \Big].
\end{align*}
This gives the design of $p_2$ in Algorithm~\ref{alg:filter}  which has the desired feature: its expectation on $D$ equals zero since $m(x)$ is an orthonormal basis in $L^2(\Rn, D)$. Yet, we remark that the degree of $p_2$ is as large as $2d$, which leads to a heavy-tailed distribution even for uncorrupted data; and thus the nasty adversary may inject comparably heavy-tailed data. In Lemma~\ref{lem:p2-properties}, we show  that the $L^2$-norm of $p_2$ on $D$ is upper bounded by $\rho_2 = O(d^d)$; thus the threshold $\kappa$ is proportional to $\rho_2$. The additional multiplicative factor in $\kappa$, $(c_0 \log\frac{1}{\eta} + c_0 d)^d \cdot \eta$, is the maximum amount that those $\eta$-fraction of heavy-tailed outliers can deteriorate the restricted Frobenius norm without appearing quite different from uncorrupted samples. In other words, with this scaling of $\kappa$, if the outliers were to deviate our estimate significantly, they would  trigger the filtering condition.

Now we give intuition on Step~3 of Algorithm~\ref{alg:filter}. We can use standard results on Gaussian tail bound of polynomials \cite{janson1997gaussian} to show that
\begin{equation*}
\Pr\big( \abs{p_2(D)} \geq t \big) \leq \exp( - (t/\rho_2)^{1/d} /  c_0 ), \ \forall t > 0.
\end{equation*}
By uniform convergence \cite{vapnik1971uniform}, the above implies a low frequency of the event $\abs{p_2(x)} \geq t$ on a set of uncorrupted samples (provided that the sample size is large enough; see Part~\ref{con:p2(S)=p2(D)} of Definition~\ref{def:good}.). On the other hand, the empirical average of $p_2$ on the input instance set $S'$ (which equals $\fronorm{(\Sigma-I)_U}$) is large. Thus, there must be some threshold $t$ such that $\abs{p_2(x)} \geq t$ occurs with a nontrivial frequency, and this is an indicator of being outliers. In Step~3 of Algorithm~\ref{alg:filter}, the nontrivial frequency is set as a constant factor of the one of uncorrupted samples~--~it is known that this suffices to produce a cleaner instance set; see e.g. \citet{diakonikolas2016robust}. To further guarantee a bounded running time, we show that it suffices to find a $t$ in $(0, 4k\gamma^2)$, thanks to the pruning step (see Lemma~\ref{lem:p2-properties}).

It is worth mentioning that our primary treatment on attribute efficiency lies in  applying uniform convergence to derive the low frequency event. In fact, since the size of $U$ is at most $4k^2$, it is possible to show that the VC-dimension of the class  $\sparseclass$ that $p_2$ resides is $O(s \log n^d)$, with $s = 4k^2$.

%However, due to the existence of outliers and the observation that the restricted Frobenius norm is larger than it supposed to be, it is possible to show t

\subsection{Termination}

Lastly, we describe the case that Algorithm~\ref{alg:main} terminates and output $u$ at Step~8. Due to the selection of $U$, it is possible to show that $\fronorm{(\Sigma-I)_U} \leq \kappa$ implies $v\trans \Sigma v \leq \kappa + 1$ for all $2k$-sparse unit vector $v$, i.e. the maximum eigenvalue of $\Sigma$ on all $2k$-sparse directions is as small as $\kappa+1$. This in turn implies that the variance caused by corrupted samples is well-controlled. Therefore, we output the empirical Chow vector. We note that Algorithm~\ref{alg:main} outputs $u$ which is the empirical one followed by a hard thresholding operation. This ensures that $u$ is $k$-sparse, the same sparsity level as $\Chow_{f^*}$. More importantly, since we are only guaranteed with a small maximum eigenvalue on $2k$-sparse directions, it is likely that on the full direction, the maximum eigenvalue could be very large, which would fail to certify a good approximation to $\Chow_{f^*}$. In other words, had we not applied the hard thresholding operation, the empirical estimate $\EXP_{(x, y)\sim \bar{S}'_l}[y \cdot m(x)]$ could be far away from the target Chow vector.

The maximum number of iterations, $l_{\max}$, comes from our analysis on the progress of the filtering step: we will show in Section~\ref{sec:guarantee} that each time Algorithm~\ref{alg:filter} is invoked, a noticeable fraction of outliers will be removed while most clean samples are retained, thus after at most $l_{\max}$ iterations, the restricted Frobenius norm must be less than  $\kappa$.

%Before Algorithm~\ref{alg:main} returns at phase $l'$, the empirical covariance matrix must have a large Frobenius norm on its top entries in a sense that \filter applies. By repeatedly applying Lemma~\ref{lem:large-fro-norm}, we have $\Delta(S, S'_{l+1}) \leq \Delta(S, S'_l) - \alpha_2$ for all $l \leq l' - 1$. By telescoping, we have $\Delta(S, S'_{l_{\max}+1}) \leq \Delta(S, S'_1) - \alpha_2 \cdot l_{\max}$. Since for any $S,S'$, $\Delta(S, S')\geq0$, and $\Delta(S, S_1') \leq 2 \epsilon$, the algorithm must return within $l_{\max} = \big\lceil \frac{2\epsilon}{\alpha_2} \big\rceil$ phases. By simple calculation on given paramters, we conclude that the algorithm is going to terminate within $\red{O\Big(\big(dk\cdot\log\frac{\epsilon k}{\alpha\delta_0}\big)^d\Big)}$ phases.

%% file: analysis.tex
\section{Performance Guarantees}\label{sec:guarantee}

Our analysis of filtering relies on the existence of a good set $S \subset \Rn$ and shows that Algorithm~\ref{alg:filter} strictly reduces the distance between the corrupted set and $S$ every time it is invoked by Algorithm~\ref{alg:main}, until the termination condition is met (Theorem~\ref{thm:large-fro}). We then show that the output of Algorithm~\ref{alg:main} must be close to the Chow vector of the underlying PTF (Theorem~\ref{thm:small-fro}), and this occurs within $l_{\max}$ phases (Theorem~\ref{thm:main-chow}). Then, a black-box application of the algorithmic result from \citet{TTV09,DDFS14,diakonikolas2018learning} leads to PAC guarantees of a PTF that is reconstructed from our estimated Chow vector (Theorem~\ref{thm:main}).

We will phrase our results in terms of some deterministic conditions on $S$. Let $S_{|\Xgamma} := S \cap \Xgamma$ and $D_{|\Xgamma}$ be the distribution $D$ conditioned on the event $x \in \Xgamma$.

\begin{definition}[Good set]\label{def:good}
Given $\epsilon \in (0, 1)$,  $\delta \in (0, 1)$, and concept class $\calH_{d, K}$, we say an instance set $S \subset \Rn$ is a good set if all the following properties hold simultaneously and uniformly over all $p_1 \in \sparselinearclass$ ($k$ is given in Lemma~\ref{lem:K-k}), all $p_2 \in \sparseclass$ with $s=4k^2$, and all $t > 0$:
\begin{enumerate}

\item $S_{|\Xgamma} = S$;\label{con:S-gamma=S}

\item $\abs{\Pr(p_1(S) > t) - \Pr(p_1(D) > t) } \leq \alpha_1$;\label{con:p1(S)=p1(D)}

\item $\abs{\Pr\big(p_1(S_{|X_{\gamma }})>t \big) - \Pr\big(p_1(D_{|X_{\gamma }})>t \big)} \leq \alpha_1$;\label{con:p1(S-gamma)=p1(D-gamma)}

\item $\abs{\EXP_{x \sim S} \big[f(x) \cdot p_1(x) \big] - \EXP_{x \sim D }\big[ f(x) \cdot p_1(x) \big]} \leq \alpha_1'$;\label{con:p1-uniform-convergence}

\item $\abs{\Pr(p_2(S) > t) - \Pr(p_2(D) > t) } \leq \alpha_2$;\label{con:p2(S)=p2(D)}

\item $\abs{\Pr\big(p_2(S_{|X_{\gamma }})>t \big) - \Pr\big(p_2(D_{|X_{\gamma }})>t \big)} \leq \alpha_2$;\label{con:p2(S-gamma)=p2(D-gamma)}

\item $\abs{ \EXP[p_2(S)] - \EXP[p_2(D)] } \leq \alpha_2'$,\label{con:p2-uniform-convergence}

\end{enumerate}
where  $\alpha_1= \frac{\epsilon}{k \gamma^2}$, $\alpha_1' = \epsilon/6$, $\alpha_2 =  \frac{\epsilon}{4k\gamma^2}$, $\alpha_2' = \epsilon$.
\end{definition}

We show that for a set of instances independently drawn from $D$, it is indeed a good set. Note that this gives the sample size at Step~\ref{step:main-sample} of Algorithm~\ref{alg:main}.
\begin{proposition}\label{prop:sample-comp}
Let $S$ be a set of $C \cdot \frac{d^{5d} {K^{4d}}}{\epsilon^2} \log^{5d}\big( \frac{n d }{\epsilon\delta} \big)$ instances drawn independently from $D$. Then with probability $1-\delta$, $S$ is a good set.
\end{proposition}

%In the rest of the analysis, we   condition on the existence of such a good set $S$.

\subsection{Analysis of \textsc{SparseFilter}}

Recall in Definition~\ref{def:nasty} that the nasty adversary first draws $S$ according to $D$ and annotates it with $f^*$ to obtain $\bar{S} \subset \Rn \times \{-1, 1\}$. Then it replaces an $\eta$ fraction with malicious samples to generate the sample set $\bar{S}'$ that is returned to the learner. Denote by $\Delta(S, S')$ the symmetric difference between $S$ and $S'$ normalized by $\abs{S}$, i.e.
\begin{equation}
\Delta(S, S') := \frac{\abs{S\backslash S'} + \abs{S' \backslash S}}{\abs{S}}.
\end{equation}
By definition, it follows that $\Delta(S, S') \leq 2\eta$. The following theorem is the primary characterization of the performance of our filtering approach (Algorithm~\ref{alg:filter}).

\begin{theorem}\label{thm:large-fro}
Consider Algorithm~\ref{alg:filter}. Assume that $\fronorm{ (\Sigma-I)_{U}} > \kappa$ and  there exists a good set $S$ such that $\Delta(S, S') \leq 2\eta$.  Then there exists a threshold $t$ that satisfies Step~3. In addition, the output $S''$ satisfies $\Delta(S, S'') \leq \Delta(S, S') - \frac{\epsilon}{4k \gamma^2}$.
\end{theorem}

%We sketch the proof of the theorem.
We show this theorem by contradiction: had we been unable to find such $t$, the tail bound at Step~3 would have implied small expectation of $p_2$ on $S'$. As discussed in Section~\ref{subsec:alg:filter}, the polynomial $p_2$ is chosen such that $\fronorm{(\Sigma - I)_U} = \EXP[p_2(S')]$; this in turn suggests that we would contradict the condition  $\fronorm{(\Sigma - I)_U} > \kappa$ when $\kappa$ is properly chosen. 

%It is useful to decompose $S'$ into two subsets: $S'\cap S$ and $S' \backslash S$. Since $S'\cap S$ is a subset of the good set $S$, its contribution to $\EXP[p_2(S')]$ can be well-controlled. Thus, it must be the case that a large $\EXP[p_2(S')]$ is resulted from abnormal values of $p_2(x)$ on $x \in S'\backslash S$.

Formally, let $\beta'(\tau, d, \rho) := 2 \cdot \rho \cdot \big(  c_0 \log\frac{1}{\tau} +  c_0 \cdot d/2 \big)^{d/2} \cdot \tau$ and $\gamma_2 := 4k^2 \gamma^2$. We have:
\begin{lemma}\label{lem:exp-p2-E}
Consider Algorithm~\ref{alg:filter}. Assume that $\fronorm{ (\Sigma-I)_{U}} > \kappa$ and  there exists a good set $S$ such that $\Delta(S, S') \leq 2\eta$. Let $E := S' \backslash S$. If there does not exist a threshold $t>0$ that satisfies Step~3, then 
\begin{equation*}
\frac{\abs{E} }{ \abs{S'} } \sup_{p_2 \in \sparseclass} \EXP[ \abs{p_2(E)} ] \leq 7(1 + \frac{1}{c}) \cdot \big[ \beta'(\eta, 2d, \rho_2) + \alpha_2 \gamma_2 \big].
\end{equation*}
\end{lemma}

\begin{lemma}\label{lem:exp-p2-L}
Consider Algorithm~\ref{alg:filter}. Assume that there exists a good set $S$ with $\Delta(S, S') \leq 2\eta$. Let $L := S \backslash S'$. We have 
\begin{equation*}
\frac{\abs{L} }{ \abs{S} } \sup_{p_2 \in \sparseclass} \EXP[ \abs{p_2(L)} ] \leq 2(1 + \frac{1}{c}) \big[ \beta'(\eta, 2d, \rho_2) + \alpha_2 \gamma_2 \big].
\end{equation*}
\end{lemma}
Now observe that $\abs{S'} \cdot \fronorm{ (\Sigma-I)_U} = \abs{S'} \cdot \EXP[p_2(S')] = \abs{S} \cdot \EXP[p_2(S)] + \abs{E} \cdot \EXP[p_2(E)] - \abs{L} \cdot \EXP[p_2(L)]$.
For the right-hand side, we can roughly think of $\EXP[p_2(S)] \approx \EXP[p_2(D)]$ which can be bounded as $D$ is Gaussian. This combined with Lemma~\ref{lem:exp-p2-E} and Lemma~\ref{lem:exp-p2-L} can establish the existence of $t$. We then use a general result that is implicit in prior filter-based algorithms \cite{diakonikolas2016robust}: given the existence of $t$, there must be a nontrivial fraction of the instances in $S'$ that can be filtered; see also Lemma~\ref{lem:progress} where we provide a generic proof. This establishes Theorem~\ref{thm:large-fro}; see Appendix~\ref{sec:app:proof-large-fro} for the full proof.

\subsection{Analysis of termination}

Let $\beta_{\tau} = 2 \big( c_0 \log\frac{1}{\tau} + c_0 d \big)^d \cdot \tau$ for some parameter $\tau \in (0, 1)$. The following theorem shows that whenever the termination condition is met, i.e. $\fronorm{(\Sigma-I)_U} \leq \kappa$, the output must be close to the target Chow vector.

\begin{theorem}\label{thm:small-fro}
Consider Algorithm~\ref{alg:main}. If at some phase $l$ we have $\fronorm{(\Sigma-I)_{U}} \leq \kappa$ and $\Delta(S, S'_{l}) \leq 2\eta$ for some good set $S$, then the following holds for the output $u$:
\begin{equation*}
\twonorm{u - \Chow_{f^*}} \leq \frac{192}{c^2} \sqrt{\eta(\beta_{\eta} + \beta_{\epsilon}) }+ \frac{\epsilon}{2}.
\end{equation*}
\end{theorem}

We note that the upper bound seems not depending on $\kappa$~--~this is because $\kappa \leq O(\beta_{\epsilon})$. To show the theorem, we will first prove that the deviation of the expectation of $y \cdot p_1(x)$ between $\bar{S}'$ and $\bar{S}$ is small, and then apply Part~\ref{con:p1-uniform-convergence} of Definition~\ref{def:good} to establish the closeness to the expectation on $D$. To obtain the first deviation bound, we observe that it is almost governed by the expectation on $S\backslash S'$ and on $S'\backslash S$. The former is easy to control since it is a subset of the good set $S$. We show that the latter is also bounded since the termination condition implies a small variance on all sparse directions of the covariance matrix $\Sigma$ that is computed on $S'$; this suggests that the contribution from the corrupted instances cannot be large. See Appendix~\ref{sec:app:proof-main-alg} for the proof.

\subsection{Main results}

\begin{theorem}[Chow vector estimation]\label{thm:main-chow}
The following holds for Algorithm~\ref{alg:main}. Given any target error rate $\epsilon \in (0, 1)$ and failure probability $\delta \in (0, 1)$, Algorithm~\ref{alg:main} runs in at most $l_{\max} = \frac{4\eta k }{\epsilon} \cdot \big( C_1 d \cdot \log\frac{nd}{\epsilon\delta} \big)^d + 1$ phases, and outputs a $k$-sparse vector $u$ such that with probability at least $1-\delta$, 
\begin{equation*}
\twonorm{ u - \Chow_{f^*}} \leq \frac{192}{c^2} \sqrt{\eta(\beta_{\eta} + \beta_{\epsilon}) }+ \frac{\epsilon}{2}.
\end{equation*}
In addition, Algorithm~\ref{alg:main} runs in $O(\poly((nd)^d, 1/\epsilon))$ time.
\end{theorem}
\begin{proof}[Proof sketch]
In view of Proposition~\ref{prop:sample-comp} and Step~\ref{step:main-sample} of Algorithm~\ref{alg:main}, there is a good set $S$ such that $\Delta(S, S') \leq 2\eta$.
We will inductively show the invariant $\Delta(S, S_{l+1}') \leq \Delta(S, S_l') - \frac{\epsilon}{4k\gamma^2}$ before  Algorithm~\ref{alg:main} terminates. In fact, by Part~\ref{con:S-gamma=S} of Definition~\ref{def:good}, it follows that no instances in $S$ will be pruned at Step~\ref{step:main-prune}. Thus, $\Delta(S, S'_1) \leq \Delta(S, S') \leq 2\eta$. If $\fronorm{(\Sigma-I)_U} > \kappa$, then Theorem~\ref{thm:large-fro} implies that we can obtain $S'_2$ such that $\Delta(S, S_{2}') \leq \Delta(S, S_1') - \frac{\epsilon}{4k\gamma^2}$. By induction, we can show that such progress holds for any phase $l$ before the termination condition is met. Since the symmetric difference is non-negative, the algorithm must terminate within the claimed $l_{\max}$ phases, upon when the output is guaranteed to be close to $\Chow_{f^*}$ in view of Theorem~\ref{thm:small-fro}. See Appendix~\ref{sec:app:proof-main-alg} for the full proof.
\end{proof}

Lastly, the algorithmic results from \citet{TTV09,DDFS14,diakonikolas2018learning} state that as long as $u$ is $\epsilon$-close to $\Chow_{f^*}$ under the $\ell_2$-norm, it is possible to construct a PTF $\hat{f}$ in time $(\frac{n}{\epsilon})^{O(d)}$ that has misclassification error of $O(d \cdot \epsilon^{1/(d+1)})$. This gives our main result on PAC guarantees (see Appendix~\ref{sec:app:PAC} for the proof).

\begin{theorem}[PAC guarantees]\label{thm:main}
There exists an algorithm $\calA$ such that the following holds. Given any $\epsilon_0 \in (0, 1)$, failure probability $\delta \in (0, 1)$, and the concept class $\calH_{d, K}$, it draws  $C \cdot \frac{d^{5d} K^{4d}}{\epsilon_0^2} \cdot \log^{5d}\big( \frac{nd}{\epsilon_0\delta} \big)$ samples from $\oraclexy(\eta)$ and outputs a PTF $\hat{f}$ such that with probability at least $1-\delta$, 
\begin{equation*}
\Pr_{x\sim D}( \hat{f}(x) \neq f^*(x)) \leq c_1 \cdot d \cdot \Big(  \sqrt{\eta(\beta_{\eta} + \beta_{\epsilon_0}) }+ {\epsilon_0} \Big)^{\frac{1}{d+1}}.
\end{equation*}
In particular, for any target error rate $\epsilon \in (0, 1)$, by setting $\epsilon_0 = \frac{\epsilon^{d+1}}{c_2 \cdot  d^{2d}}$, we have $\Pr_{x\sim D}( \hat{f}(x) \neq f^*(x)) \leq \epsilon$
provided $\eta \leq \frac12 \epsilon_0$. Moreover, the algorithm runs in $(nd/\epsilon)^{O(d)}$ time.
\end{theorem}

%% file: con.tex
\section{Conclusion}\label{sec:con}

We studied the important problem of attribute-efficient PAC learning of low-degree PTFs. We showed that for the class of sparse PTFs where the concept depends only on a subset of its input attributes, it is possible to design an efficient algorithm that PAC learns the class with sample complexity poly-logarithmic in the dimension, even in the presence of the nasty noise. In addition, the noise tolerance of our algorithm is dimension-independent, and matches the best known result established for learning of non-sparse PTFs. %To this end, we gave a structural result which translates the attribute sparsity to a sparsity pattern of the Chow vector. We then developed an attribute-efficient Chow vector estimator which leverages a crucially chosen polynomial to filter corrupted samples.

%% file: app.tex
We summarize a few useful results and list reserved hyper-parameters in Appendix~\ref{sec:app:notaion}; they will be frequently used in our analysis. We provide proofs for results in Section~\ref{sec:intro} and Section~\ref{sec:setup} in Appendix~\ref{sec:app:prem}. Appendix~\ref{sec:app:general} collects statistical results on the concept classes of interest, which are used in Appendix~\ref{sec:app:proof-large-fro} and Appendix~\ref{sec:app:proof-main-alg} to establish guarantees on Algorithm~\ref{alg:filter} and Algorithm~\ref{alg:main}, respectively. We assemble all pieces and prove the main result, Theorem~\ref{thm:main}, in Appendix~\ref{sec:app:PAC}.

\section{Summary of Useful Facts and Reserved Hyper-Parameters}\label{sec:app:notaion}

We will often need the condition that $\Delta(S, S') \leq 2\eta$, which implies
\begin{equation}
(1-2\eta) \abs{S} \leq \abs{S'} \leq \abs{S}.
\end{equation}
In particular, when $\eta \in [0, \frac12 - c]$ for some absolute constant $c \in (0, \frac12]$, we have
\begin{equation}\label{eq:S-S'}
\abs{S'} \leq \abs{S} \leq \frac{1}{1-2\eta} \abs{S'} \leq \big(1+ \frac1c \cdot \eta\big) \abs{S'}.
\end{equation}
The above two inequalities also imply
\begin{equation}\label{eq:L-E}
\abs{S' \backslash S} \leq 2\eta \abs{S} \leq \frac{2\eta}{1-2\eta} \abs{S'} \leq \frac{\eta}{c} \abs{S'} \quad \text{and}\quad \abs{S \backslash S'} \leq 2\eta \abs{S} \leq \frac{2\eta}{1-2\eta} \abs{S'} \leq \frac{\eta}{c} \abs{S'}.
\end{equation}

It is known that for any vector $u$,
\begin{equation}
\onenorm{u} \leq \sqrt{\zeronorm{u}} \cdot \twonorm{u}.
\end{equation}
The above will often be applied together with Holder's inequality:
\begin{equation}\label{eq:holder}
\abs{u \cdot v} \leq \onenorm{u} \cdot \infnorm{v} \leq \sqrt{\zeronorm{u}} \cdot \twonorm{u} \cdot \infnorm{v}.
\end{equation}

\begin{fact}\label{fact:exp-by-prob}
Let $Z$ be a positive random variable. Then $\EXP[Z] = \int_{0}^{\infty} \Pr(Z > t) \dif t$.
\end{fact}

\begin{fact}[Tail bound of Gaussian polynomials \cite{janson1997gaussian}]\label{fact:tail-bound}
Let $D$ be the standard Gaussian distribution $\calN(0, I_{n\times n})$.
There exists an absolute constant $c_0 > 1$ such that the following tail bound holds for all degree-$d$ polynomials $p: \Rn \rightarrow \R$:
\begin{equation*}
\Pr_{x \sim D} \big( \abs{p(x) - \EXP[p(D)]} \geq t \sqrt{\Var[p(D)]} \big) \leq \exp( -  t^{2/d} / c_0), \ \forall\ t > 0.
\end{equation*}
In particular, if $p$ is such that $\EXP[p(D)] = 0$, we have
\begin{equation*}
\Pr_{x \sim D} \big( \abs{p(x) } \geq t \Ltwonorm{p} \big) \leq \exp( -  t^{2/d} / c_0).
\end{equation*}
\end{fact}

%\subsection{Reserved Constants}\label{sec:constants}

%The constants $c_1$, $c_2$, and $c_4$ are defined in Lemma~\ref{lem:var-S}, Lemma~\ref{lem:var-L}, and Lemma~\ref{lem:var-E-large-fro} respectively. The constants $c_5$ and $c_1$ are set out in Lemma~\ref{lem:exp-p2-S} and Lemma~\ref{lem:exp-p2-E-restate} respectively. The constant $C_1$ appearing in $\kappa$ is chosen as $C_1 =  (1 + \frac{1}{2c}) c_5 + c_1$, where the constant $c$ is such that $\epsilon \in [0, \frac12 - c]$.

\subsection{Reserved Hyper-Parameters}

Recall that $\epsilon \in (0, 1)$ is the noise rate, $\delta \in (0, 1)$ is the failure probability, $d$ is the degree of the PTFs.
Denote by $X_{\gamma } = \{x \in \Rn: \infnorm{m(x)} \leq \gamma  \}$ the instances of interest. Given an instance set $S \subset \Rn$, let $S_{|X_{\gamma }} = S \cap X_{\gamma }$. For a distribution $D$ supported on $\Rn$, let $D_{|X_{\gamma }}$ be the distribution $D$ conditioned on the event that $x \in X_{\gamma }$. 
\begin{itemize}

\item $\beta(\tau, d, \rho) =  \rho^2 \cdot  \big(  c_0 \log\frac{1}{\tau} +  c_0 d \big)^d \cdot \tau$, which upper bounds $\int_{0}^{\infty}t \cdot \min\{\tau,Q_{\rho,d}(t)\}\dif t$ for $Q_{d, \rho}(t) = \exp(- (t/\rho)^{2/d}/ c_0 )$; see Lemma~\ref{lem:beta};

\item $\beta'(\tau, d, \rho) = 2 \cdot \rho \cdot \big(  c_0 \log\frac{1}{\tau} +  c_0 \cdot d/2 \big)^{d/2} \cdot \tau $, which upper bounds $\int_{0}^{\infty} \min\{\tau,Q_{\rho,d}(t)\}\dif t$; see Lemma~\ref{lem:beta-2};

%\item $\abs{S} = C \cdot \frac{d^{5d} {K^{4d}}}{\epsilon^2} \log^{5d}\big( \frac{n d }{\epsilon\delta} \big)$; see Proposition~\ref{prop:sample-comp-restate};

\item $\gamma  = \big( c_0 \log\frac{\abs{S} \cdot (n+1)^d}{\delta_{\gamma}}  \big)^{d/2}  = \big( C_1 d \cdot  \log\frac{nd}{\epsilon\delta} \big)^{d/2}$, which upper bounds $\max_{x\in S} \infnorm{m(x)}$  with probability $1-\delta_{\gamma}$ for $S$ drawn from $D$; see Lemma~\ref{lem:gamma-restate} and Definition~\ref{def:good};

\item $\gamma_1 = \sqrt{2k} \gamma$, which upper bounds $\abs{p_1(x)}$ for $p_1 \in \sparselinearclass$ and $x \in \Xgamma$; see Lemma~\ref{lem:p1-properties};

\item $\gamma_2 = 2\sqrt{s} \gamma^2$, which upper bounds $\abs{p_2(x)}$ for $p_2 \in \sparseclass$ and $x \in \Xgamma$; see Lemma~\ref{lem:p2-properties};

\item $\rho_2 ={C_2} \cdot d^{\frac34} \cdot \big( c_0 d \big)^{d}$, which upper  bounds $\Ltwonorm{p_2}$ for $p_2 \in \sparseclass$; see Lemma~\ref{lem:p2-properties};

%\item $\alpha_1  =\frac{\epsilon}{\gamma_1^2}$, $\alpha_2  =\frac{\epsilon}{\gamma_2}$;  see Definition~\ref{def:good}.
\end{itemize}

\section{Omitted Proofs from Section~\ref{sec:intro} and Section~\ref{sec:setup}}\label{sec:app:prem}

\subsection{Proof of Fact~\ref{fact:baseline}}

\begin{proof}
It is known from \citet{maass1994fast} that in the absence of noise, $\calH_{d, K}$ can be PAC learned efficiently by using linear programming to find a concept that fits all the samples since in this case, empirical risk minimization with $0/1$-loss is equivalent to solving a linear program. The number of samples is at least $N_{\delta} := C_0 \cdot \frac{1}{\epsilon^2}( K^d \log n + \log \frac{1}{\delta})$ due to uniform convergence theory \cite{blumer1989learn} and the VC-dimension of $\calH_{d, K}$.  Then Theorem~12 of \citet{kearns1988learning} shows that when $\eta \leq \frac{1}{N_{1/2}} \log N_{1/2}$, it is possible to learn the same concept class by using  $2N_{\delta}^2 \log\frac{1}{\delta} \cdot N_{\delta} = 2N_{\delta}^3 \log\frac{1}{\delta}$ samples. Substituting $N_{\delta}$ gives the result.
\end{proof}

\subsection{Proof of Lemma~\ref{lem:K-k}}

\begin{proof}
By definition, we know that there exists $q \in \polyclass$ and $\Omega \subset [n]$ with $\abs{\Lambda}\leq K$, such that $f(x) = \sign( q(x_{\Lambda}))$. Let $\compset{\Lambda} := [n]\backslash\Lambda$. Since we choose Hermite polynomials as the basis, we have that $m_i(x) = m_i(x_{\Lambda}) \cdot m_i(x_{\compset{\Lambda}})$ where we define $m_i(x_{\compset{\Lambda}}) = 1$ if $m_i(x)$ does not depend on $x_{\compset{\Lambda}}$.

We calculate the $i$-th element of the Chow vector of $f$ as follows:
\begin{align}
&\ \EXP_{x\sim D}[f(x) \cdot {m}_i(x)] \notag\\
=&\ \EXP_{x\sim D}[\sign( q(x_{\Lambda})) \cdot m_i(x)] \notag\\
=&\  \EXP_{x\sim D}[ \sign( q(x_{\Lambda})) \cdot m_i(x_{\Lambda}) \cdot m_i(x_{\overline{\Lambda}})] \label{eq:product_basis}\\
=&\  \EXP_{x\sim D}[ \sign( q(x_{\Lambda})) \cdot m_i(x_{\Lambda})] \cdot \EXP_{x\sim D}[m_i(x_{\overline{\Lambda}})] \label{eq:indepedent_coordinates}\\
=&\  \EXP_{x\sim D}[ \sign( q(x_{\Lambda})) \cdot m_i(x_{\Lambda})] \cdot 0 = 0 \label{eq:zero_means}
\end{align}
as long as $m_i$ depends on some elements in $\overline{\Lambda}$. Equivalently, the above is non-zero for all $m_i$ that depends only on $\Lambda$. Note that there are at most
\begin{equation}
\sum_{i=0}^d K^i
\leq \begin{cases}
d+1,\ &\text{if}\ K=1,\\
\frac{K^{d+1}-1}{K-1} \leq 2 K^d,\ &\text{if}\ K \geq 2
\end{cases}
\end{equation}
such $m_i$'s. This gives the desired sparsity bound.
\end{proof}

\subsection{The basis of monomials}

As a complementary discussion to Lemma~\ref{lem:K-k}, we also give derivation for the sparsity of the Chow parameters under the basis of monomial polynomials.

\begin{lemma}\label{lem:K-k-monomial}
Let $f$ be a $K$-sparse degree-$d$ PTF. Then $\Chow_f$ is a $k$-sparse vector under the basis of monomial polynomials, where $k \geq \big(\frac{n}{d}\big)^{\lfloor\frac{d}{2}\rfloor}$.
\end{lemma}
\begin{proof}
Consider the same setting as that in Lemma~\ref{lem:K-k}, except that the basis is now under monomial polynomials. Since multivariate monomials are constructed by the production of univariate monomials, the $i$-th element of the Chow vector of $f$ can be written as
\begin{equation*}
\EXP_{x\sim D}[f(x) \cdot {m}_i(x)] = \EXP_{x\sim D}[ \sign( q(x_{\Lambda})) \cdot m_i(x_{\Lambda})] \cdot \EXP_{x\sim D}[m_i(x_{\overline{\Lambda}})] .
\end{equation*}
However, now the term $\EXP_{x\sim D}[m_i(x_{\overline{\Lambda}})]$ equals zero only when $m_i(x_{\overline{\Lambda}})$ includes at least one univariate monomial $x_j^{\ell}$ for some $j\in\overline{\Lambda}$ where $\ell\in\mathbb{Z}_{+}$ is an odd integer. For $K\leq\frac{n}{2}, d\geq2$, the total number of non-zero elements in $\Chow_f$ is at least
\begin{align*}
\sum_{j=1}^{\lfloor\frac{d}{2}\rfloor}\binom{n-K+j-1}{j} \geq \sum_{j=1}^{\lfloor\frac{d}{2}\rfloor} \bigg(\frac{n-K+j-1}{j}\bigg)^{j} \geq \bigg(\frac{n-K+\lfloor\frac{d}{2}\rfloor-1}{\lfloor\frac{d}{2}\rfloor}\bigg)^{\lfloor\frac{d}{2}\rfloor} \geq \bigg(\frac{n}{d}\bigg)^{\lfloor\frac{d}{2}\rfloor},
\end{align*}
which depends polynomially in the dimension $n$, making it undesirable in the attribute-efficient learning.
\end{proof}

\subsection{Other choices of polynomial basis}\label{sec:app:basis}

As mentioned in Section~\ref{subsec:techniques}, the choice of appropriate basis that demonstrates sparsity structure of the Chow parameters can go well beyond that of Hermite polynomials. More generally, under the assumption that $D$ is a product distribution (Eq.\eqref{eq:indepedent_coordinates}), we require the basis to be a product basis (so Eq.\eqref{eq:product_basis} holds) with zero mean under the distribution $D$ (so Eq.\eqref{eq:zero_means} holds). To this end, there exist other choices of basis under different distributional assumptions. For example, under the uniform distribution over $[-1,1]^n$, the multivariate Legendre polynomials also form an appropriate basis. Another necessary property of $D$ in our analysis is the finiteness of moments up to order $4d$, so that we can obtain tail bound for any degree-$4d$ polynomial; this is needed to establish Lemma~\ref{lem:p2-properties}.

\section{General Statistical Results}\label{sec:app:general}

Recall that we assume $D$ is the standard Gaussian $\calN(0, I_{n\times n})$ in this paper.

\subsection{Tail bound on $\infnorm{m(x)}$}

The tail bound of Fact~\ref{fact:tail-bound} implies the following upper bound on the magnitude of $m(x)$.

\begin{lemma}[Restatement of Lemma~\ref{lem:gamma}]\label{lem:gamma-restate}
The following holds for all $t > 1$:
\begin{align*}
\Pr_{x \sim D}\big( \infnorm{m(x)} \geq t \big) &\leq (n+1)^d \exp(-t^{2/d} / c_0),\\
\Pr_{S \sim D}\big( \max_{x\in S} \infnorm{m(x)} \geq t \big) &\leq \abs{S}  \cdot (n+1)^d \cdot \exp(-t^{2/d} / c_0).
\end{align*}
In particular, with probability at least $1- \delta_{\gamma}$, we have $\max_{x \in S} \infnorm{m(x)} \leq \gamma $ where $\gamma :=  \big( c_0 \log\frac{\abs{S}(n+1)^d}{\delta_{\gamma}}  \big)^{d/2}$. When $\delta_{\gamma} = \frac{\epsilon^2 \delta}{64\rho_2^2}$ and $\abs{S} = C \cdot \frac{d^{5d}K^{4d}}{\epsilon^2} \log^{5d}\big(\frac{nd}{\epsilon\delta}\big)$, we have $\gamma = \big(C_1d \cdot \log\frac{nd}{\epsilon\delta}\big)^{d/2}$.
\end{lemma}
\begin{proof}
Denote $M = (n+1)^d$ the dimension of the vector $m(x)$.
Let $m_i(x)$ be the $i$-th element of $m(x)$, where $1 \leq i \leq M$. We note that $m_1(x) \equiv 1$. Now for any $ i \neq 1$, since $m(x)$ is orthonormal in $L^2(\Rn, D)$, we have $\EXP[m_i(D)] = 0$ and $\Ltwonorm{m_i} = 1$. By Fact~\ref{fact:tail-bound}, we have
\begin{equation*}
\Pr_{x \sim D}( \abs{m_i(x)} \geq t ) \leq \exp(-t^{2/d} / c_0).
\end{equation*}
Taking the union bound over all $i \in \{2, \dots, M\}$ gives
\begin{equation*}
\Pr_{x \sim D}\big( \max_{2 \leq i \leq M} \abs{m_i(x)} \geq t \big) \leq (M-1) \exp(-t^{2/d} / c_0) \leq M \exp(-t^{2/d} / c_0).
\end{equation*}
Note that for all $t > 1$, we have
\begin{equation}
\Pr_{x \sim D}\big( \infnorm{m(x)} \geq t \big) \leq M \exp(-t^{2/d} / c_0).
\end{equation}

Now for $S$ being a set of independent draws from $D$, we have by union bound that
\begin{equation}
\Pr_{S \sim D}\big( \max_{x\in S} \infnorm{m(x)} \geq t \big) \leq \abs{S}  \cdot M \cdot \exp(-t^{2/d} / c_0).
\end{equation}
The proof is complete.
\end{proof}

\subsection{$\beta(\tau, d, \rho), \beta'(\tau, d, \rho)$}

%\Jie{thank you for the great reference!}

%Parameter $\beta\geq \int_{0}^{\infty}t\min\{\epsilon,Q_d(t)\}\dif t$ quatifies the maximum amount of variation in $\EXP(p^2(D))$ by shifting an $\epsilon$-probability mass. For standard Gaussian distribution $\calN(0,I)$, $Q_d(t) = \exp(-t^{2/d} / c_0)$. 

\begin{lemma}\label{lem:beta}
Let $Q_{d, \rho}(t) = \exp(-(t/\rho)^{2/d} / c_0)$ where $\rho$ is independent of $t$. Then $\int_{0}^{\infty}t \cdot \min\{\tau, Q_{d, \rho}(t)\}\dif t \leq \beta(\tau, d, \rho)$ where $\beta(\tau, d, \rho) := \rho^2 \cdot \big( c_0 \log\frac{1}{\tau} + c_0  d \big)^d \cdot \tau$.
\end{lemma}
\begin{proof}
Denote $\rho_0 = c_0 \cdot \rho^{2/d}$.
Let $t_0 = ({\rho_0}\log\frac{1}{\tau})^{d/2}$, which satisfies $\tau=Q_{d, \rho}(t_0)$. It follows that
\begin{align*}
\int_{0}^{\infty}t \cdot \min\{\tau, Q_{d, \rho}(t)\}\dif t &= \int_{0}^{t_0} t\cdot\tau \dif t + \int_{t_0}^{\infty} t\cdot e^{-t^{2/d} / \rho_0} \dif t \\
&\stackrel{\zeta_1}{=} \frac12t_0^2 \tau + \frac{d \cdot \rho_0^d}{2}\int_{t_0^{2/d} / \rho_0}^{\infty}  e^{-u} \cdot u^{d-1} \dif u \\
&\stackrel{\zeta_2}{=} \frac12t_0^2 \tau + \frac{d \cdot \rho_0^d}{2} \cdot \Gamma(d, t_0^{2/d} / \rho_0)\\
&\stackrel{\zeta_3}{\leq} \frac12t_0^2 \tau + \frac{d \cdot \rho_0^d}{2} \cdot \exp(- t_0^{2/d} / \rho_0) \cdot ( t_0^{2/d} / \rho_0 + d)^{d-1}\\
&\stackrel{\zeta_4}{=} \frac{1}{2} \cdot \tau \cdot (\rho_0 \log\frac{1}{\tau})^d + \frac{d \cdot \rho_0^d}{2} \cdot \tau \cdot \big( \log\frac{1}{\tau} + d \big)^{d-1}\\
&\leq \big( \rho_0 \log\frac{1}{\tau} + \rho_0 d \big)^d \cdot \tau.
\end{align*}
In the above, $\zeta_1$ follows from the change of variables $u := t^{2/d}/\rho$, $\zeta_2$ follows from the definition of upper incomplete gamma function, $\zeta_3$ follows from Lemma~\ref{lem:incomplete-gamma}, $\zeta_4$ follows from the setting of $t_1$, and the last step follows from simple algebraic relaxation. The result follows by replacing $\rho_0$ with $c_0 \cdot \rho^{2/d}$.
\end{proof}

\begin{corollary}\label{coro:beta-p1}
The following holds:
\begin{equation*}
\sup_{p_1 \in \sparselinearclass} \int_{0}^{\infty} t \cdot \min\{ \tau, \Pr( \abs{p_1(D)} > t) \} \leq \beta(\tau, d, \sqrt{2}).
\end{equation*}
\end{corollary}
\begin{proof}
We will use the tail bound from Lemma~\ref{lem:p1-properties}.
Let $t_0 = \sqrt{2} (c_0 \log \frac{1}{\tau})^{d/2}$. Since $c_0 > 1$, we have $t_0 > \sqrt{2}$. By Lemma~\ref{lem:p1-properties}, we also have $\Pr(\abs{p_1(D)} > t_0 ) \leq \tau$. Therefore, we can write
\begin{equation*}
\int_{0}^{\infty} t \cdot \min\{ \tau, \Pr(p_1(D) \geq t) \} \dif t = \int_{0}^{t_0} t \cdot \tau \dif t + \int_{t_0}^{\infty} t \cdot e^{-(t/\sqrt{2})^{2/d} / c_0} \dif t.
\end{equation*}
Using the same proof as Lemma~\ref{lem:beta} gives the desired result.
\end{proof}

\begin{corollary}\label{coro:beta-p2}
The following holds:
\begin{equation*}
\sup_{p_2 \in \sparseclass} \int_{0}^{\infty} t \cdot \min\{ \tau, \Pr( \abs{p_2(D)} > t) \} \leq \beta(\tau, 2d, \rho_2).
\end{equation*}
\end{corollary}
\begin{proof}
This follows immediately from the tail bound on $\Pr(\abs{p_2(D)} > t)$ in Lemma~\ref{lem:p2-properties} and Lemma~\ref{lem:beta}.
\end{proof}

The following lemma provides another type of bound.

\begin{lemma}\label{lem:beta-2}
{Let $Q_{d, \rho}(t) = \exp(-(t / \rho)^{2/d} / c_0)$ where $\rho$ is independent of $t$. Then $\int_{0}^{\infty} \min\{\tau,Q_{d, \rho}(t)\}\dif t \leq \beta'(\tau, d, \rho)$ where $\beta'(\tau, d, \rho) := 2  \rho \cdot \big( c_0 \log\frac{1}{\tau} +  c_0 \cdot \frac{d}{2} \big)^{d/2} \cdot \tau$.}
\end{lemma}
\begin{proof}
Denote $\rho_0 = c_0 \cdot \rho^{2/d}$.
Let $t_0 = ({\rho_0}\log\frac{1}{\tau})^{d/2}$, which satisfies $\tau=Q_{d, \rho}(t_0)$. It follows that
\begin{align*}
\int_{0}^{\infty}t \cdot \min\{\tau,Q_{d, \rho}(t)\}\dif t &= \int_{0}^{t_0} \tau \dif t + \int_{t_0}^{\infty}  e^{-t^{2/d} / \rho_0} \dif t \\
&\stackrel{\zeta_1}{=} t_0 \tau +  \rho_0^{d/2} \int_{t_0^{2/d} / \rho_0}^{\infty}  e^{-u} \cdot u^{d/2-1} \dif u \\
&\stackrel{\zeta_2}{=} t_0 \tau +  \rho_0^{d/2} \cdot \Gamma(d/2, t_0^{2/d} / \rho)\\
&\stackrel{\zeta_3}{\leq} t_0 \tau +  \rho_0^{d/2} \cdot \exp(- t_0^{2/d} / \rho_0) \cdot ( t_0^{2/d} / \rho_0 + d/2)^{d/2-1}\\
&\stackrel{\zeta_4}{=}  \tau \cdot (\rho_0 \log\frac{1}{\tau})^{d/2} + \rho_0^{d/2} \cdot \tau \cdot \big( \log\frac{1}{\tau} + \frac{d}{2} \big)^{d/2-1}\\
&\leq 2\rho_0^{d/2} \cdot \big(  \log\frac{1}{\tau} +  \frac{d}{2} \big)^{d/2} \cdot \tau.
\end{align*}
In the above, $\zeta_1$ follows from the change of variables $u := t^{2/d}/\rho$, $\zeta_2$ follows from the definition of upper incomplete gamma function, $\zeta_3$ follows from Lemma~\ref{lem:incomplete-gamma}, $\zeta_4$ follows from the setting of $t_0$, and the last step follows from simple algebraic relaxation. The result follows by noting $\rho_0 = c_0 \cdot \rho^{2/d}$.
\end{proof}

\begin{lemma}[Claim~3.11 of \cite{diakonikolas2018list}]\label{lem:incomplete-gamma}
Consider the upper incomplete gamma function $\Gamma(s, x) = \int_{x}^{\infty} t^{s-1} e^{-t} \dif t$. We have $\Gamma(s, x) \leq e^{-x} \cdot (x+s)^{s-1}$ for all $s \geq 1$ and $x \geq 0$.
\end{lemma}

%\begin{align*}
%\int_{0}^{\infty}t \cdot \min\{\epsilon,Q_d(t)\}\dif t &= \int_{0}^{t_1} t\cdot\epsilon \dif t + \int_{t_1}^{\infty} t\cdot e^{-c_0t^{2/d}} \dif t \\
%&= \epsilon\cdot\frac12t_1^2 + \int_{t_1^{2/d} / c_0}^{\infty} u^{\frac{d}{2}}e^{-u}\cdot\frac{d}{2}u^{\frac{d}{2}-1}\dif u \\
%&=\frac12\epsilon\cdot\Big(\frac{1}{c_0}\log\frac{1}{\epsilon}\Big)^{d} + \frac{d}{2}\int_{\frac{1}{c_0}\log\frac{1}{\epsilon}}^{\infty} u^{d-1}e^{-u}\dif u \\ 
%&\leq \frac12\epsilon\cdot\Big(\frac{1}{c_0}\log\frac{1}{\epsilon}\Big)^{d} + \frac{d}{2}\Gamma\Big(d,\frac{1}{c_0}\log\frac{1}{\epsilon}\Big) \\
%&\leq \frac12\epsilon\cdot\Big(\frac{1}{c_0}\log\frac{1}{\epsilon}\Big)^{d} + \frac{d}{2}\cdot e^{-\frac{1}{c_0}} \cdot \epsilon\Big(d+\frac{1}{c_0}\log\frac{1}{\epsilon}\Big)^{d-1}
%\end{align*}

\subsection{Concept class: basic properties}

Recall that $m(x)$ is the vector of all Hermite polynomials on $\Rn$ with degree at most $d$. Note that $m(x)$ has $(n+1)^d$ elements. We defined two classes of polynomials:
\begin{equation}
\sparselinearclass := \{ p: \Rn \rightarrow \R: p(x) = \inner{v}{{m}(x)}, v \in \R^{(n+1)^d}, \twonorm{v} = 1, \zeronorm{v} \leq k \},
\end{equation}
and
\begin{equation}
\sparseclass := \{ p: \Rn \rightarrow \R: p = \langle A_U, m m\trans - I \rangle, A \in \mathbb{S}^{(n+1)^d}, U\trans = U, \zeronorm{U} \leq s, \fronorm{A_U} = 1 \},
\end{equation}
where
\begin{equation}
\mathbb{S}^{(n+1)^d} = \{ A: A \in \R^{(n+1)^d \times (n+1)^d}, A\trans = A\}.
\end{equation}
Note  that the polynomials in $\sparseclass$ have degree at most $2d$, and can be represented as a linear combination of at most $s$ elements of the form $m_i(x) m_j(x)$.

We collect a few basic properties of the sparse polynomial classes.

\begin{lemma}\label{lem:p1-properties}
For all $p_1 \in \sparselinearclass$, the following holds: 
\begin{itemize}

\item Deterministic property: $\abs{p_1(x)} \leq \sqrt{2k} \cdot \infnorm{{m}(x)}$; in particular, $\abs{p_1(x)} \leq \gamma_1$ for all $x \in X_{\gamma }$, where $\gamma_1 := \sqrt{2k} \gamma$;

\item Distributional property: for all $t \geq \sqrt{2}$,
\begin{equation*}
\Pr\big( \abs{p_1(D)} \geq t \big)   \leq e^{-(t/\sqrt{2})^{2/d} / c_0}.
\end{equation*}

\end{itemize}
\end{lemma}
\begin{proof}
By Holder's inequality, we have
\begin{equation*}
\abs{p_1(x)} = \abs{v \cdot {m}(x)} \leq  \sqrt{\zeronorm{v}} \cdot \twonorm{v} \cdot \infnorm{m(x)} \leq \sqrt{2k} \cdot 1 \cdot \infnorm{m(x)},
\end{equation*}
where the first inequality follows from \eqref{eq:holder} and the second inequality follows from the fact that $p_1 \in \sparselinearclass$.

To show the tail bound, we will decompose $v = (v_1, \tilde{v})$ and $m(x) = (1, \tilde{m}(x))$ so that $\EXP[\tilde{m}(D)] = 0$. In this way, we have $p_1(x) = v_1 + \tilde{v} \cdot \tilde{m}(x) = v_1 + \twonorm{\tilde{v}} \cdot q(x)$, where $q(x) := \bar{v} \cdot \tilde{m}(x)$ and $\bar{v} := \tilde{v} / \twonorm{\tilde{v}}$.

Observe that $\EXP[q(D)] = 0$ and $\Var[q(D)] = 1$. By Fact~\ref{fact:tail-bound}, for all $t>0$,
\begin{equation}\label{eq:tmp:q(D)}
\Pr\big( \abs{q(D)} \geq t  \big) \leq e^{- t^{2/d} / c_0}.
\end{equation}
By simple calculation, we can show that
\begin{equation*}
\abs{p_1(x)}  = \sqrt{ (v_1 + \twonorm{\tilde{v}} \cdot q(x))^2 } \leq \sqrt{2} \cdot \sqrt{v_1^2 + \twonorm{\tilde{v}}^2 \cdot q^2(x) } \leq \sqrt{2} \abs{q(x)}
\end{equation*}
for $q(x) \geq 1$. Thus, for all $t \geq \sqrt{2}$, we have
\begin{equation}
\Pr( \abs{p_1(D)} \geq t ) \leq \Pr( \sqrt{2} \abs{q(D)} \geq t) = \Pr( \abs{q(D)} \geq t/\sqrt{2})  \leq e^{-(t/\sqrt{2})^{2/d} / c_0},
\end{equation}
where the last step follows from \eqref{eq:tmp:q(D)}.
\end{proof}

\begin{lemma}\label{lem:p2-properties}
For all $p_2 \in \sparseclass$, the following holds: 
\begin{itemize}

\item Deterministic property: $\abs{p_2(x)} \leq 2\sqrt{s} \cdot  \infnorm{m(x)}^2$; in particular, $\abs{p_2(x)} \leq \gamma_2$ for all $x \in X_{\gamma }$, where $\gamma_2 := 2\sqrt{s} \cdot \gamma^2$;

\item Distributional properties: $\EXP[p_2(D)] = 0$, $\Ltwonorm{p_2} \leq \rho_2$ where $\rho_2 := {C_2} \cdot d^{\frac34} \cdot \big( c_0 d \big)^{d}$, and
\begin{equation*}
\Pr\big( \abs{p_2(D)} \geq t \big) \leq \exp( - (t/\rho_2)^{1/d} /  c_0 ), \ \forall t > 0.
\end{equation*}

\end{itemize}

\end{lemma}
\begin{proof}
By the Cauchy-Schwartz inequality,
\begin{equation*}
\abs{p_2(x)} \leq \fronorm{A_U} \cdot \fronorm{\big( m(x) m(x)\trans - I \big)_U }  \leq \sqrt{\zeronorm{U}} \cdot \infnorm{m(x)}^2 + \fronorm{I_U} \leq \sqrt{s} \big( \infnorm{m(x)}^2 + 1 \big),
\end{equation*}
where in the second inequality we use the fact that each entry of the matrix $m(x)m(x)\trans$ takes the form $m_i(x) m_j(x)$ whose magnitude is always upper bounded by $\infnorm{m(x)}^2$, and there are at most $\zeronorm{U}$ such entries. Since $m_1(x) = 1$, we always have $1 \leq \infnorm{m(x)}^2$. Thus $\abs{p_2(x)} \leq 2\sqrt{s} \cdot \infnorm{m(x)}^2$.

Now we show the distributional properties. Since $m(x)$ is a complete orthonormal basis in $L^2(\Rn, D)$, it follows that
\begin{equation}
\EXP[p_2(D)] = \inner{A_U}{I - I} = 0.
\end{equation}

Now we bound $\Ltwonorm{p_2}$, which equals $\sqrt{ \EXP[p_2^2(D)] }$. Recall that $A_U$ is symmetric  with $\fronorm{A_U} = 1$, and thus can be written as
\begin{equation}
A_U = V\trans \Lambda V,
\end{equation}
for some orthonormal matrix $V$ and diagonal matrix $\Lambda$ with $\fronorm{\Lambda}=1$. Let $q(x) = V \cdot m(x)$. Observe that $\EXP[q(D)] = 0$ and $\EXP[q(D) q(D)\trans] = I$. Then
\begin{equation}\label{eq:tmp:exp-p2(D)}
\EXP[p_2^2(D)] = \Var\Big[ q(D)\trans \Lambda q(D) \Big] = \Var\Big[ \sum_i \Lambda_{ii} q_i^2(D) \Big] \leq  \sum_i \Lambda_{ii}^2 \Var[q_i^2(D)],
\end{equation}
where $\Lambda_{ii}$ denotes the $i$-th diagonal element of $\Lambda$ and $q_i(x)$ denotes the $i$-th component of the vector-valued function $q(x)$, and the last step follows from the fact that $\EXP[q_i(D) \cdot q_j(D)] = 0$ for $i\neq j$ and $\EXP[q_i(D)] = 0$ for all $i$. It thus remains to upper bound $\Var[q_i^2(D)]$.

By the definition of variance, we have
\begin{equation*}
\Var[q_i^2(D)] = \EXP[q_i^4(D)] - ( \EXP[q_i^2(D)] )^2 = \EXP[q_i^4(D)] - 1 \leq C_2^2 \cdot d^{\frac32} \cdot \big( c_0 d \big)^{2d},
\end{equation*}
where we applied Lemma~\ref{lem:var-p2-D} in the last step. Plugging it into \eqref{eq:tmp:exp-p2(D)}, we get
\begin{equation}
\EXP[p_2^2(D)] \leq C_2^2 \cdot d^{\frac32} \cdot \big( \frac{2d-1}{c_0 e} \big)^{2d-1} \sum_i \Lambda_{ii}^2 = C_2^2 \cdot d^{\frac32} \cdot \big( c_0 d \big)^{2d},
\end{equation}
where the equality follows from the construction that $\fronorm{\Lambda} = 1$ and $\sum_i \Lambda_{ii}^2 = \fronorm{\Lambda}^2$. The proof is complete by noting that $\Ltwonorm{p_2} = \sqrt{\EXP[p_2^2(D)]}$.
\end{proof}

\begin{lemma}\label{lem:var-p2-D}
Let $v$ be a unit vector and $m(x)$ be a collection of orthonormal polynomials of degree at most $d$ in $L^2(\Rn, D)$. Then $\EXP[(v \cdot m(D) )^4] \leq C_2^2 \cdot d^{\frac32} \cdot \big( c_0 d \big)^{2d}$ for some sufficiently large constant $C_2 > 0$.
\end{lemma}
\begin{proof}
Let $p(x) = v \cdot m(x)$ and $\rho = c_0 \cdot 2^{1/d}$.
We bound the desired expectation by using Fact~\ref{fact:exp-by-prob} with the tail bound in Lemma~\ref{lem:p1-properties}:
\begin{align*}
\EXP[p^4(D)] &= \int_{0}^{\infty} \Pr( p^4(D) > t) \dif t\\
&= \int_{0}^{\sqrt{2}} \Pr( p^4(D) > t) \dif t + \int_{\sqrt{2}}^{\infty} \Pr( p^4(D) > t) \dif t\\
&\leq \sqrt{2} + 4 \int_{0}^{\infty} t^3 \cdot \Pr( \abs{p(D)} > t) \dif t\\
&\leq \sqrt{2} + 4 \int_{0}^{\infty} t^3 \cdot e^{-  t^{2/d} / \rho} \dif t\\
&= \sqrt{2} + \frac{d}{2} \cdot  \rho^{2d-1} \cdot \int_{0}^{\infty} t^{2d-1} e^{-t} \dif t\\
&\stackrel{\zeta_1}{=} \sqrt{2} + \frac{d}{2} \cdot \rho^{2d-1} \cdot (2d-1)!\\
&\stackrel{\zeta_2}{\leq} \sqrt{2} + \frac{d}{2} \cdot \rho^{2d-1} \cdot \sqrt{2\pi (2d-1)} \big( \frac{2d-1}{e} \big)^{2d-1} \cdot e^{\frac{1}{12(2d-1)}}
\end{align*}
where  $\zeta_1$ follows from the known fact on the value of the Gamma function, and $\zeta_2$ follows from the Stirling's approximation. The result follows by noting that $\rho = c_0 \cdot 2^{1/d}$ and choosing a large enough constant $C_2$.
\end{proof}

\subsection{Concept class: uniform convergence and sample complexity}

\begin{lemma}\label{lem:vc-dim}
The VC-dimension of the class $\calH^1_{n,d,2k} := \{ h: x \mapsto \sign(p_1(x)), p_1 \in \sparselinearclass \}$ is $O(d \cdot k \log n)$, and that of the class $\calH^2_{n,d,s} := \{ h: x \mapsto  \sign(p_2(x)), p_2 \in \sparseclass \}$ is $O(d \cdot s \log n)$. 
\end{lemma}
\begin{proof}
For the class $\calH^1_{n,d,k}$, we can consider the class of polynomials in $ \sparselinearclass$ with a fixed support set. It is easy to see that the VC-dimension of such class is $k+1$. Now note that the number of the choices of such support set is
\begin{equation*}
\sum_{i=0}^{k} \binom{ (n+1)^{d} }{ i } \leq \big( \frac{e (n+1)^{d} }{k} \big)^k.
\end{equation*}
The concept class union argument states that for any $\mathcal{H} = \cup_{i=1}^M \mathcal{H}_i$, the VC dimension of $\mathcal{H}$ is upper bounded by $O(\max\{V, \log M + V \log\frac{\log M}{V} \})$, where $V$ is an upper bound on the VC dimension of all $\mathcal{H}_i$.
Thus, the VC-dimension of $\calH^1_{n,d,2k}$ is $O(d \cdot k \log n)$ by algebraic calculations. 

Likewise, for $\calH^2_{n,d,s}$, 
we can first fix the support $U \subset [(n+1)^d] \times [(n+1)^d]$ in the representation of $p_2$. Let $\calP(U)$ be the class of polynomials in $\sparseclass$ with the fixed $U$. It is easy that the VC-dimension of $\calP(U)$ is $s+1$. Now note that the number of choices of $U$ is 
\begin{equation*}
\sum_{i=0}^{s} \binom{ (n+1)^{2d} }{ i } \leq \big( \frac{e (n+1)^{2d} }{s} \big)^s.
\end{equation*}
Using the same argument gives that the VC-dimension of $\calH^2_{n,d,s}$ is $O(d \cdot s \log n)$.
\end{proof}

\begin{proposition}[Restatement of Prop.~\ref{prop:sample-comp}]\label{prop:sample-comp-restate}
Let $S$ be a set of $C \cdot \frac{d^{5d} {K^{4d}}}{\epsilon^2} \log^{5d}\big( \frac{n d }{\epsilon\delta} \big)$ instances drawn independently from $D$, where $C>0$ is a sufficiently large constant. Then with probability $1-\delta$, $S$ is a good set in the sense of Definition~\ref{def:good}.
\end{proposition}
\begin{proof}
By Lemma~\ref{lem:gamma}, our setting of $\delta_{\gamma}$ and $\abs{S}$, it follows that with probability at least $1-\delta_{\gamma}$, we must have $S_{| \Xgamma} = S$. This proves Part~\ref{con:S-gamma=S}. From now on, we condition on this event happening.

We note that by the classical VC theory \cite{anthony1999neural} and our VC-dimension upper bound in Lemma~\ref{lem:vc-dim}, Parts~\ref{con:p1(S)=p1(D)},~\ref{con:p1(S-gamma)=p1(D-gamma)},~\ref{con:p2(S)=p2(D)},~\ref{con:p2(S-gamma)=p2(D-gamma)} in Definition~\ref{def:good} all hold with probability $1-\delta/4$ as far as 
\begin{equation}\label{eq:tmp:S-1}
\abs{S} \geq C \cdot \big( \frac{\vcdim}{\alpha^2} \cdot  \log\frac{4 \cdot \vcdim}{\alpha \delta}   \big),
\end{equation}
where $\vcdim := \max\{d \cdot k \log n, d \cdot s \log n\} = d \cdot k^2 \log n$, and $\alpha := \min\{ \alpha_1, \alpha_2\} \leq \frac{\epsilon}{2 k \gamma^2}$.

%\Jie{this is an argument that only applies to a given $S$ (i.e. input data). Do we really need to establish uniform style result like \cite{diakonikolas2018learning}? }

We now show Part~\ref{con:p1-uniform-convergence}. For $x \in \Xgamma$, we have $\abs{m_i(x)} \leq \gamma$ with certainty. Therefore, by Hoeffding's inequality for bounded random variables, we have that
\begin{equation*}
\Pr\big( \abs{ \EXP_{S_{|X_{\gamma }}}[ f(x)m_i(x)] - \EXP_{D_{|\Xgamma}}[ f(x) m_i(x)] } > t \big) \leq 2\exp\big( - \frac{ \abs{S} t^2}{4 \gamma^2} \big).
\end{equation*}
Therefore, taking the union bound over $i$ we obtain that if 
\begin{equation}\label{eq:tmp:S-2}
\abs{S} \geq \frac{32 k \gamma^2}{(\alpha_1')^2}  \cdot \log\frac{16 (n+1)^d}{\delta},
\end{equation}
then with probability at least $1 - \delta/8$, we have
\begin{equation*}
\max_{1 \leq i \leq (n+1)^d} \abs{ \EXP_{S_{|X_{\gamma }}}[ f(x)m_i(x)] - \EXP_{D_{|\Xgamma}}[ f(x) m_i(x)] } \leq \frac12 \cdot \frac{\alpha_1'}{\sqrt{2k}}.
\end{equation*}
Now we observe that for any $p_1 \in \sparselinearclass$, we have $p_1(x) = \inner{v}{m(x)}$ with $\onenorm{v} \leq \sqrt{k}$. Thus
\begin{align}\label{eq:uni-p1-1}
&\ \abs{\EXP_{x \sim S_{|\Xgamma} } \big[f(x) \cdot p_1(x) \big] - \EXP_{x \sim D_{|\Xgamma} }\big[ f(x) \cdot p_1(x) \big]} \notag \\
=&\ \abs{v \cdot \big( \EXP_{x \sim S_{|\Xgamma} } \big[f(x) \cdot m(x) \big] - \EXP_{x \sim D_{|\Xgamma} }\big[ f(x) \cdot m(x) \big] \big)} \notag \\
\leq&\ \sqrt{k} \cdot \infnorm{ \EXP_{x \sim S_{|\Xgamma} } \big[f(x) \cdot m(x) \big] - \EXP_{x \sim D_{|\Xgamma} }\big[ f(x) \cdot m(x) \big] } \notag\\
\leq&\ \frac{1}{2} \alpha_1'.
\end{align}
On the other hand,
recall that for any $p_1 \in \sparselinearclass$, $\Ltwonorm{p_1} = 1$ in view of Lemma~\ref{lem:p1-properties}. Thus Lemma~\ref{lem:exp-D-D'} tells that
\begin{equation*}
\sup_{f: \abs{f}\leq 1, p_1 \in \sparselinearclass}  \abs{ \EXP_{x\sim D}[ f(x) \cdot p(x)] - \EXP_{x \sim  D_{|X_{\gamma}}} [ f(x) \cdot p(x)] } \leq 4  \sqrt{ \Pr_{x\sim D}(x \notin X_{\gamma }) } \leq   4\sqrt{\delta_{\gamma}},
\end{equation*}
where the last step follows from Lemma~\ref{lem:gamma-restate}. Since we set $\delta_{\gamma}$ such that $\delta_{\gamma} \leq \frac{ (\alpha_1')^2}{64}$, we have
\begin{equation}\label{eq:uni-p1-2}
\sup_{f: \abs{f}\leq 1, p_1 \in \sparselinearclass}  \abs{ \EXP_{x\sim D}[ f(x) \cdot p(x)] - \EXP_{x \sim  D_{|X_{\gamma}}} [ f(x) \cdot p(x)] }  \leq  \frac12 \alpha_1'.
\end{equation}
Part~\ref{con:p1-uniform-convergence} follows from applying triangle inequality on \eqref{eq:uni-p1-1} and \eqref{eq:uni-p1-2}, and  the conditioning $S_{|\Xgamma} = S$.

Lastly, we show Part~\ref{con:p2-uniform-convergence}. We note that for any fixed index $(i, j) \in [(n+1)^d] \times  [(n+1)^d]$, $\sup_{x \in \Xgamma} \abs{m_{i}(x) m_{j}(x)}  \leq \gamma^2$ holds with certainty. Therefore, Hoeffding's inequality for bounded random variable tells that
\begin{equation*}
\Pr\big( \abs{ \EXP_{S_{|X_{\gamma }}}[ m_i(x) m_j(x)] - \EXP_{D_{|\Xgamma}}[  m_i(x) m_j(x)] } > t \big) \leq 2\exp\big( - \frac{ \abs{S} t^2}{4 \gamma^4} \big).
\end{equation*}
Thus, by taking the union bound over all choices of $(i, j)$, we obtain that with probability $1-\delta/8$,
\begin{equation}
\max_{(i, j) \in [(n+1)^d] \times  [(n+1)^d] } \abs{ \EXP_{S_{|X_{\gamma }}}[ m_i(x) m_j(x)] - \EXP_{D_{|\Xgamma}}[  m_i(x) m_j(x)] } \leq \frac12 \cdot \frac{ \alpha_2'}{\sqrt{s}}
\end{equation}
as far as 
\begin{equation}\label{eq:tmp:S-3}
\abs{S} \geq \frac{16 s \gamma^4}{(\alpha_2')^2} \cdot  \log\frac{16 (n+1)^{2d}}{\delta}.
\end{equation}

Now, for any $p_2 \in \sparseclass$, we have $p_2(x) = \inner{A_U}{m(x) m(x)\trans} - \inner{A_U}{I}$. Thus,
\begin{align}\label{eq:tmp:p2-uni-1}
&\ \abs{ \EXP_{S_{|X_{\gamma }}}[ p_2(x)] - \EXP_{D_{|\Xgamma}}[  p_2(x)] } \notag\\
=&\ \abs{ \sum_{(i, j) \in U} A_{ij} \big( \EXP_{S_{|X_{\gamma }}}[ m_i(x) m_j(x) ] - \EXP_{D_{|\Xgamma}}[  m_i(x) m_j(x) ] \big) } \notag\\
\leq&\ \sqrt{\zeronorm{U}} \cdot \max_{(i, j) \in [(n+1)^d] \times  [(n+1)^d] } \abs{ \EXP_{S_{|X_{\gamma }}}[ m_i(x) m_j(x)] - \EXP_{D_{|\Xgamma}}[  m_i(x) m_j(x)] } \notag\\
\leq&\ \sqrt{s} \cdot \frac12 \cdot \frac{\alpha_2'}{\sqrt{s}} \notag\\
=&\ \frac{1}{2} \alpha_2'.
\end{align}
On the other hand, by Lemma~\ref{lem:p2-properties}, we have $\Ltwonorm{p_2} \leq \rho_2$ for all $p_2 \in \sparseclass$. Thus, Lemma~\ref{lem:exp-D-D'} tells that
\begin{equation*}
\sup_{p_2 \in \sparseclass} \abs{ \EXP_{x\sim D}[ p_2(x)] - \EXP_{x \sim  D_{|X_{\gamma}}} [ p_2(x)] } \leq 4 \rho_2 \sqrt{ \Pr_{x\sim D}(x \notin X_{\gamma }) } \leq  4\rho_2 \sqrt{ \delta_{\gamma}}.
\end{equation*}
Since we set $\delta_{\gamma}$ such that $\delta_{\gamma} \leq \frac{(\alpha_2')^2}{64 \rho_2^2}$, we have
\begin{equation}\label{eq:tmp:p2-uni-2}
\sup_{p_2 \in \sparseclass} \abs{ \EXP_{x\sim D}[ p_2(x)] - \EXP_{x \sim  D_{|X_{\gamma}}} [ p_2(x)] } \leq 4 \rho_2 \sqrt{ \Pr_{x\sim D}(x \notin X_{\gamma }) } \leq  \frac12 \alpha_2'.
\end{equation}
Part~\ref{con:p2-uniform-convergence} follows from applying triangle inequality on \eqref{eq:tmp:p2-uni-1} and \eqref{eq:tmp:p2-uni-2}, and the conditioning $S_{|\Xgamma} = S$.

Observe that by the union bound, all these parts hold simultaneously with probability at least $1 - \delta_{\gamma} - \delta/4 - \delta/8 - \delta/8 \geq 1 - \delta$ since we set $\delta_{\gamma} \leq \delta/2$. In addition, to satisfy all the requirements on the sample size involved in all parts, i.e. \eqref{eq:tmp:S-1}, \eqref{eq:tmp:S-2}, and \eqref{eq:tmp:S-3}, we need
\begin{equation}
\abs{S} \geq C' \cdot \big( \frac{\gamma^4 \cdot k^4 \cdot \log n}{\epsilon^2} \log\frac{\gamma k  n^d d }{\epsilon \delta} \big),
\end{equation}
for some large enough constant $C' > 0$. Our setting on $\abs{S}$ follows by plugging the setting of $\gamma$ in Definition~\ref{def:good} into the above equation and noting that $k \leq \max\{d+1,  2K^n\}$. The proof is complete.
%\begin{equation}
%\abs{S} \geq C \cdot \frac{k^4 \cdot \log n }{\epsilon^2} \cdot \big(  \log\frac{k^4 d^{4d} n^{2d}}{\epsilon\delta} \big)^{4d},
%\end{equation}
%where $C > 0$ is a large enough constant and $C_1 > 0$ is the constant defined in $\gamma$. In particular, when $K = 1$, we have $k = d+1$ and thus we can set
%\begin{equation}\label{eq:S-1}
%\abs{S} \geq C \cdot \frac{d^{5d}}{\epsilon^2} \big( \log\frac{nd}{\epsilon\delta} \big)^{5d}.
%\end{equation}
%If $K > 1$, we have $k = 2K^d$ and thus we can set
%\begin{equation}\label{eq:S-2}
%\abs{S} \geq C \cdot \frac{d^{4d} \cdot K^{4d} }{\epsilon^2} \cdot \Big( \log\frac{nd}{\epsilon\delta} \Big)^{5d}.
%\end{equation}
%It is not hard to see that our setting of $\abs{S}$ satisfies both \eqref{eq:S-1} and \eqref{eq:S-2}.
\end{proof}

\begin{lemma}[Total variation distance]\label{lem:exp-D-D'}
Assume that $\Pr_{x\sim D}(x \in X_{\gamma}) \geq \frac12$ and let $\rho > 0$ be a finite real number. The following holds uniformly for all functions $f$ and $p$ satisfying $f: \Rn \rightarrow [-1, 1]$ and $\Ltwonorm{p} \leq \rho$:
\begin{equation*}
\abs{ \EXP_{x\sim D}[ f(x) \cdot p(x)] - \EXP_{x \sim  D_{|X_{\gamma}}} [ f(x) \cdot p(x)] } \leq 4 \rho \sqrt{ \Pr_{x\sim D}(x \notin X_{\gamma }) }.
\end{equation*}
\end{lemma}
\begin{proof}
Denote $z(x) = f(x) \cdot p(x)$.
Let $\ind{X_{\gamma }^c}(x)$ be the indicator function which outputs $1$ if $x \notin X_{\gamma }$ and $0$ otherwise. By simple calculation, we have
\begin{equation*}
\EXP_{x \sim D}[ z(x)] = \Pr_{x\sim D}(x \in X_{\gamma }) \cdot \EXP_{D_{|X_{\gamma }}} [z(x)] + \EXP_D[ z(x) \cdot \ind{X_{\gamma }^c}(x)],
\end{equation*}
namely,
\begin{equation}
\EXP_{D_{|X_{\gamma }}} [z(x)] = \frac{ \EXP_D[z(x)] }{\Pr_D(x \in X_{\gamma })} - \frac{ \EXP_D[ z(x) \cdot \ind{X_{\gamma }^c}(x)] }{\Pr_D(x \in X_{\gamma })}.
\end{equation}
Therefore,
\begin{align}\label{eq:tmp:uni-1}
\abs{ \EXP_{D_{|X_{\gamma }}} [z(x)] - \EXP_D[z(x)] } &= \abs{ \frac{ \Pr_D(x \notin X_{\gamma }) \cdot \EXP_D[z(x)] }{\Pr_D(x \in X_{\gamma })} - \frac{ \EXP_D[ z(x) \cdot \ind{X_{\gamma }^c}(x)] }{\Pr_D(x \in X_{\gamma })}  } \notag\\
&\leq \frac{ \Pr_D(x \notin X_{\gamma }) }{\Pr_D(x \in X_{\gamma })} \cdot \abs{\EXP_D[z(x)]} + \frac{ \abs{ \EXP_D[ z(x) \cdot \ind{X_{\gamma }^c}(x)] } }{\Pr_D(x \in X_{\gamma })} \notag\\
&\leq \frac{1}{1/2} \cdot \Pr_D(x \notin X_{\gamma }) \cdot \abs{\EXP_D[z(x)]} + \frac{ \abs{ \EXP_D[ z(x) \cdot \ind{X_{\gamma }^c}(x)] } }{ 1/2 },
\end{align}
where we applied the condition $\Pr_D(x \in X_{\gamma }) \geq 1/2$ in the last step.

Observe that
\begin{equation}\label{eq:tmp:uni-2}
\abs{\EXP_D[z(x)]} \leq \sqrt{\EXP_D[z^2(x)]} \leq \sqrt{\EXP_D[p^2(x)]} = \Ltwonorm{p} \leq \rho.
\end{equation}
On the other hand,
\begin{equation}\label{eq:tmp:uni-3}
\abs{ \EXP_D[ z(x) \cdot \ind{X_{\gamma }^c}(x)] } \leq \sqrt{ \EXP_D[z^2(x)] } \cdot \sqrt{ \EXP_D[ \ind{X_{\gamma }^c}(x) ] } \leq \rho \cdot \sqrt{ \Pr_D(x \notin X_{\gamma }) }.
\end{equation}
Combining \eqref{eq:tmp:uni-1}, \eqref{eq:tmp:uni-2}, \eqref{eq:tmp:uni-3}, and noting that any probability is always no greater than its root completes the proof.
\end{proof}

\section{Analysis of Algorithm~\ref{alg:filter}: Proof of Theorem~\ref{thm:large-fro}}\label{sec:app:proof-large-fro}

\begin{proof}[Proof of Theorem~\ref{thm:large-fro}]
Note that in view of our construction of $p_2$ in the algorithm, we have $\EXP[p_2(S')] = \fronorm{(\Sigma-I)_{U}}$. Denote $E = S' \backslash S$ and $L = S \backslash S'$. Then,
\begin{equation}\label{eq:tmp:decomp-S'}
\abs{S'} \cdot \fronorm{ (\Sigma-I)_U} = \abs{S'} \cdot \EXP[p_2(S')] = \abs{S} \cdot \EXP[p_2(S)] + \abs{E} \cdot \EXP[p_2(E)] - \abs{L} \cdot \EXP[p_2(L)].
\end{equation}
Observe that Lemma~\ref{lem:p2-properties} tells that $\EXP[p_2(D)] = 0$, which combined with Part~\ref{con:p2-uniform-convergence} of Definition~\ref{def:good} gives $\EXP[p_2(S)] \leq \alpha_2'$.
%\Shiwei{and the quadratic form of $p_2$ naturally implies $\EXP[p_2(L)] \geq 0$. This is not necessarily true. $p_2$ can be negative, otherwise the mean won't be zero. Thus, we need to bound $\abs{\abs{L}\cdot\EXP[p_2(L)]}$} 
In addition, Lemma~\ref{lem:exp-p2-L-restate} shows $\abs{L}\cdot\abs{\EXP[p_2(L)] } \leq 2(1+\frac{1}{c})\abs{S}\cdot (\beta'(\eta, 2d, \rho_2) + \alpha_2 \gamma_2 )$. Assume for contradiction that no such threshold $t$ exists. Then Lemma~\ref{lem:exp-p2-E-restate} gives $\abs{E} \cdot \abs{ \EXP[p_2(E)] } \leq 7(1+\frac{1}{c}) \abs{S'} \cdot (\beta'(\eta, 2d, \rho_2)+ \alpha_2 \gamma_2)$.
Plugging these into \eqref{eq:tmp:decomp-S'}, we obtain that
\begin{equation*}
\abs{S'} \cdot \fronorm{(\Sigma - I)_{U}} \leq \abs{S} \cdot \alpha_2' + 7(1 + \frac{1}{c}) \abs{S'} \cdot (\beta'(\eta, 2d, \rho_2)+ \alpha_2 \gamma_2) +2(1 + \frac{1}{c})\abs{S}\cdot (\beta'(\eta, 2d, \rho_2) + \alpha_2 \gamma_2 ).
\end{equation*}
Diving both sides by $\abs{S'}$ and noting that \eqref{eq:S-S'} shows $\abs{S} \leq (1 + \frac{1}{2c}) \abs{S'}$, we obtain
\begin{align*}
\fronorm{(\Sigma - I)_{U}} \leq 7 (1 + \frac1c)(1 + \frac{1}{2c}) \big( \beta'(\eta, 2d, \rho_2) + \alpha_2' + \alpha_2 \gamma_2) \leq \frac{14}{c^2} \big( \beta'(\eta, 2d, \rho_2) + \alpha_2' + \alpha_2 \gamma_2),
\end{align*}
where the last step follows since $c \in (0, \frac12]$.
Recall that we set $\alpha_2' = \epsilon$ and $\alpha_2 = \epsilon / \gamma_2$ in Definition~\ref{def:good}. Thus, the above inequality reads as
\begin{equation*}
\fronorm{(\Sigma - I)_{U}} \leq \frac{14}{c^2} \big( \beta'(\eta, 2d, \rho_2) + 2\epsilon  \big) = \kappa,
\end{equation*}
which contradicts the condition of the proposition that $\fronorm{(\Sigma - I)_{U}} > \kappa$.

 %\Shiwei{$\kappa$ updated.}

Note that the existence of such threshold $t$ combined with  Lemma~\ref{lem:progress} implies the desired progress in the symmetric difference. In particular, by combining Part~\ref{con:p2(S)=p2(D)} of Definition~\ref{def:good} and Lemma~\ref{lem:p2-properties}, we have
\begin{equation}
\Pr( p_2(S) > t ) \leq \exp(- (t/\rho_2)^{1/d}/c_0) + \alpha_2,\ \forall\ t > 0.
\end{equation}
We also just showed that there exists $t > 0$ such that
\begin{equation}
\Pr( p_2(S') > t ) \geq 6 \exp(- (t/\rho_2)^{1/d}/c_0) + 6\alpha_2.
\end{equation}
In addition, \eqref{eq:S-S'} tells $\abs{S'} \geq \frac12 \abs{S}$. Thus, Lemma~\ref{lem:progress} asserts that
\begin{equation}
\Delta(S, S'') \leq \Delta(S, S') - \exp(- (t/\rho_2)^{1/d}/c_0) - \alpha_2 \leq \Delta(S, S') - \alpha_2.
\end{equation}
This completes the proof by noting that we set $\alpha_2 = \frac{\epsilon}{\gamma_2} = \frac{\epsilon}{4k \gamma^2}$.
\end{proof}

\subsection{Auxiliary results}

\begin{lemma}[Restatement of Lemma~\ref{lem:exp-p2-E}]\label{lem:exp-p2-E-restate}
Consider Algorithm~\ref{alg:filter}. Suppose that $\Delta(S, S') \leq 2\eta$ and $\fronorm{(\Sigma-I)_{U}} > \kappa$. Let $E = S' \backslash S$. If there does not exist a threshold $t>0$ that satisfies Step~3, then 
\begin{equation*}
(\abs{E} / \abs{S'} ) \sup_{p_2 \in \sparseclass} \EXP[ \abs{p_2(E)} ] \leq 7(1 + \frac{1}{c}) \cdot \big[ \beta'(\eta, 2d, \rho_2) + \alpha_2 \gamma_2 \big].
\end{equation*}
\end{lemma}
\begin{proof}
We use Lemma~\ref{lem:general-contribution-2} to establish the result. We note that $\Delta(S, S') \leq 2\eta$ implies $\abs{E} \leq \frac{\eta}{c} \abs{S'}$ by \eqref{eq:L-E}. Since we assumed that no threshold $t$ satisfies the filtering condition, we have 
\begin{equation*}
\Pr( \abs{p_2(S')} > t) \leq 6 \exp(-(t/\rho_2)^{1/d}/ c_0) + 6 \alpha_2, \ \forall\ t > 0.
\end{equation*}
By Lemma~\ref{lem:beta-2}, we have $\int_{0}^{\infty} \min\{ \eta, \exp(-(t/\rho_2)^{1/d}/ c_0) \} \leq  \beta'(\eta, 2d, \rho_2)$. Lastly, by Lemma~\ref{lem:p2-properties}, we have $\max_{x \in S'} \abs{p_2(x)} \leq \gamma_2$. Thus, using Lemma~\ref{lem:general-contribution-2} gives the result.
\end{proof}

\begin{lemma}[Restatement of Lemma~\ref{lem:exp-p2-L}]\label{lem:exp-p2-L-restate}
Consider Algorithm~\ref{alg:filter}. Suppose that $S$ is a good set and $\Delta(S, S') \leq 2\eta$. We have 
\begin{equation*}
(\abs{L} / \abs{S} ) \sup_{p_2 \in \sparseclass} \EXP[ \abs{p_2(L)} ] \leq 2(1 + \frac{1}{c}) \big[ \beta'(\eta, 2d, \rho_2) + \alpha_2 \gamma_2 \big].
\end{equation*}
\end{lemma}
\begin{proof}
We use Lemma~\ref{lem:general-contribution-2} to establish the result. Similar to Lemma~\ref{lem:exp-p2-E-restate}, we can show that $\abs{L} \leq \frac{\eta}{c} \abs{S}$ by \eqref{eq:L-E}. By combining Lemma~\ref{lem:p2-properties} and Part~\ref{con:p2(S)=p2(D)} of Definition~\ref{def:good}, we have
\begin{equation*}
\Pr( \abs{p_2(S)} > t ) \leq \exp(-(t/\rho_2)^{1/d} / c_0) + \alpha_2,\ \forall\ t > 0.
\end{equation*}
By Lemma~\ref{lem:beta-2}, we have $\int_{0}^{\infty} \min\{ \eta, \exp(- (t/\rho_2)^{1/d}/ c_0) \} \leq  \beta'(\eta, 2d, \rho_2)$. Lastly, by Lemma~\ref{lem:p2-properties}, we have $\max_{x \in S'} \abs{p_2(x)} \leq \gamma_2$. Thus, using Lemma~\ref{lem:general-contribution-2} gives the result.
%Since $L\subset S$, similar to the proof of Lemma~\ref{lem:exp-p2-E-restate}, we have that
%\begin{equation}
%\frac{\abs{L}}{\abs{S}}\Pr(\abs{p_2(L)} > t) \leq \min\Big\{ \frac{\abs{L}}{\abs{S}},\ \Pr(\abs{p_2(S)} > t) \Big\} \leq \min\Big\{ \epsilon,\ \Pr(\abs{p_2(S)} > t) \Big\}.
%\end{equation}
%Notice that
%\begin{align}\label{eq:tmp-bound-p2-L}
%\frac{\abs{L}}{\abs{S}} \frac{1}{\twonorm{p_2}} \abs{\EXP[p_2(L)]} &= \frac{\abs{L}}{\abs{S}} \frac{1}{\Ltwonorm{p_2}} \int_{0}^{T_{\max}^2} \Pr(\abs{p_2(L)} > t) \dif t \notag\\
%&\leq \int_{0}^{T_{\max}^2} \min\Big\{ \epsilon,\ \Pr(p_2(S) > t) \Big\} \dif t \notag\\
%&\stackrel{\zeta_1}{\leq} \int_{0}^{T_{\max}^2} \min\Big\{ \epsilon,\ Q_{2d}(t) + \frac{\epsilon}{10T_{\max}^2} \Big\} \dif t \notag\\
%&\leq \int_{0}^{T_{\max}^2} \min\Big\{ \epsilon,\ Q_{2d}(t) \Big\} \dif t + \int_{0}^{T_{\max}^2}  \frac{\epsilon}{10T_{\max}^2}  \dif t \notag\\
%&\leq \beta+\frac{\epsilon}{10} ,
%\end{align}
%where $\zeta_1$ follows from Definition~\ref{def:good-set}. The proof is complete.
\end{proof}

The following lemma borrows from Lemma~2.10 of \citet{diakonikolas2018learning}; the proof is included for completeness.
\begin{lemma}\label{lem:general-contribution-2}
Let $c_0 > 0$ be an absolute constant.
Let $S_0$ be a set of instances in $\Rn$ and $S_1 \subset S_0$, with $\abs{S_1} \leq \omega_1 \tau \abs{S_0}$ for some $\omega_1, \tau > 0$. Let $p$ be such that $\Pr( \abs{p(S_0)} > t ) \leq \omega_2 \cdot Q_d(t) + \alpha_0$ for all $t \geq t_0$, where $\omega_2, Q_d(t), \alpha_0 > 0$. Assume $\max_{x \in S_0} \abs{p(x)} \leq \gamma_0$. Further assume that $\int_{0}^{\infty}  \min\{\tau, Q_d(t)\} \dif t \leq \beta_0$. Then
\begin{equation*}
( \abs{S_1} / \abs{S_0} ) \cdot \EXP[\abs{p(S_1)}] \leq \omega_1 t_0 \cdot \tau + (\omega_1+1)(\omega_2+1) \beta_0  + \alpha_0 \gamma_0.
\end{equation*}
\end{lemma}
\begin{proof}
Since $S_1 \subset S_0$, we have $\abs{S_1} \cdot \Pr( \abs{ p(S_1) } > t) \leq \abs{S_0} \cdot \Pr( \abs{ p(S_0) } > t)$. Thus,
\begin{equation}\label{eq:tmp-prob}
 \Pr( \abs{ p(S_1) } > t) \leq \min\Big\{ 1, \frac{\abs{S_0}}{\abs{S_1}} \cdot \Pr( \abs{ p(S_0) } > t) \Big\}.
\end{equation}
By Fact~\ref{fact:exp-by-prob}, we have
\begin{align*}
( \abs{S_1} / \abs{S_0} ) \cdot \EXP[\abs{ p(S_1) }] &\leq \int_{0}^{\infty}  ( \abs{S_1} / \abs{S_0} ) \Pr( \abs{ p(S_1) } > t) \dif t\\
&\stackrel{\zeta_1}{=} \int_{0}^{\gamma_0}   ( \abs{S_1} / \abs{S_0} ) \Pr( \abs{ p(S_1) } > t) \dif t\\
&\stackrel{\zeta_2}{\leq} \int_{0}^{\gamma_0}  \min\Big\{ \abs{S_1} / \abs{S_0},  \Pr( \abs{ p(S_0) } > t) \Big\} \dif t \\
&\stackrel{\zeta_3}{\leq} \int_{0}^{\gamma_0}  \min\Big\{ \omega_1 \tau,  \Pr( \abs{ p(S_0) } > t) \Big\} \dif t\\
&\stackrel{\zeta_4}{\leq} \int_{0}^{t_0} \min\{ \omega_1 \tau, 1\} \dif t + \int_{t_0}^{\gamma_0}  \min\{ \omega_1 \tau, \omega_2 \cdot Q_d(t) + \alpha_0 \} \dif t\\
&\stackrel{\zeta_5}{\leq}  \omega_1 \tau t_0 + \int_{t_0}^{\gamma_0}  \min\{ \omega_1 \tau, \omega_2 \cdot Q_d(t)  \} \dif t + \int_{t_0}^{\gamma_0}  \alpha_0 \dif t \\
&\stackrel{\zeta_6}{\leq} \omega_1 t_0 \cdot \tau + (\omega_1+1)(\omega_2+1) \int_{t_0}^{\gamma_0}  \min\{  \tau,   Q_d(t)  \} \dif t  + \alpha_0(\gamma_0 - t_0)\\
&\stackrel{\zeta_7}{\leq} \omega_1 t_0 \cdot \tau + (\omega_1+1)(\omega_2+1) \beta_0  + \alpha_0\gamma_0.
\end{align*}
In the above, $\zeta_1$ follows from the condition that $p(x) \leq \gamma_0$ for all $x \in S_1$, $\zeta_2$ follows from \eqref{eq:tmp-prob}, $\zeta_3$ uses the condition $\abs{S_1} \leq \omega_1 \tau \abs{S_0}$, $\zeta_4$ uses the condition of the tail bound of $p(S_0)$ when $t \geq t_0$, $\zeta_5$ applies elementary facts that $\min\{ \omega_1 \tau, 1\} \leq \omega_1 \tau$ and $\min\{a, b + c\} \leq \min\{a, b\} + c $ for any $c > 0$, $\zeta_6$ uses the fact that both $\frac{\omega_1}{(\omega_1+1)(\omega_2+1)}$ and $\frac{\omega_1}{(\omega_1+1)(\omega_2+1)}$ are less than $1$ for positive $\omega_1$ and $\omega_2$, and $\zeta_7$ applies the condition on the integral and uses the fact that both $\tau$ and $Q_d(t)$ are positive.
\end{proof}

The following lemma is implicit in prior works but we give a slightly more general statement; see e.g. Claim~5.13 of \citet{diakonikolas2016robust}.
\begin{lemma}\label{lem:progress}
Let $S$ and $S'$ be two instance sets with $\abs{S'} \geq \alpha \abs{S}$ for some $\alpha \in (0, 1]$. Suppose that there exists $t_0>0$ such that $\Pr( g(S) \geq t_0) \leq h_1(t_0)$, $\Pr(g(S') \geq t_0) > h_2(t_0)$, and $h_2(t_0) \geq \frac{3}{\alpha} \cdot h_1(t_0)$. Let $S'' = S' \cap \{x: g(x) \geq t_0 \}$. Then $\Delta(S, S'') - \Delta(S, S') \leq - h_1(t_0)$.
\end{lemma}
\begin{proof}
Write $E := S' \backslash S$ and $L := S \backslash S'$. Then $S' = S \cup E \backslash L$. Likewise, write $E' := S''\backslash S$ and $L' := S \backslash S''$. Then $S'' = S \cup E' \backslash L'$. Since $S'' \subset S'$, we have $E' \subset E$ and $L' \supset L$. It is not hard to see that
\begin{equation}\label{eq:tmp-lem-delta}
\Delta(S, S'') - \Delta(S, S') = \frac{\abs{E'}+\abs{L'}}{\abs{S}} - \frac{\abs{E}+\abs{L}}{\abs{S}}  = \frac{1}{\abs{S}} \cdot \big( \abs{L'\backslash L} - \abs{E\backslash E'}\big).
\end{equation}
Let $V := \{x: g(x) \geq t_0\}$. By our assumption, it follows that
\begin{equation*}
\abs{S \cap V} \leq h_1(t_0) \cdot \abs{S}, \quad \abs{S' \cap V} > h_2(t_0) \cdot \abs{S'}.
\end{equation*}
By basic set operations, we have $E \backslash E' = (S' \backslash S) \cap V = (S' \cap V) \backslash S = (S' \cap V) \backslash (S \cap V)$. Thus,
\begin{equation}\label{eq:tmp-lem-delta-1}
\abs{E \backslash E'} \geq \abs{S' \cap V} - \abs{S \cap V} \geq h_2(t_0) \cdot \abs{S'} - h_1(t_0) \abs{S} \geq \big(\alpha \cdot h_2(t_0) - h_1(t_0) \big) \abs{S}.
\end{equation}
On the other hand, $L' \backslash L = (S' \cap S) \cap V$. Thus,
\begin{equation}\label{eq:tmp-lem-delta-2}
\abs{L' \backslash L} \leq \abs{S \cap V} \leq h_1(t_0) \cdot \abs{S}.
\end{equation}
Combining \eqref{eq:tmp-lem-delta-1} and \eqref{eq:tmp-lem-delta-2}, and the condition of $h_2(t_0) \geq \frac{3}{\alpha} \cdot h_1(t_0)$, we have
\begin{equation*}
\abs{E \backslash E'} \geq 2 h_1(t_0) \cdot \abs{S} \geq \abs{L' \backslash L} + h_1(t_0) \cdot \abs{S}.
\end{equation*}
This combined with \eqref{eq:tmp-lem-delta} completes the proof.
\end{proof}

%% file: app-output.tex
\section{Performance Guarantees on the Output of Algorithm~\ref{alg:main}}\label{sec:app:proof-main-alg}

\subsection{Proof of Theorem~\ref{thm:small-fro}}

\begin{proof}[Proof of Theorem~\ref{thm:small-fro}]
We first show the following holds: 
\begin{align}\label{eq:tmp:first-eq}
\sup_{p_1 \in \sparselinearclass}\abs{\EXP_{(x, y) \sim \bar{S}'_{l} }[y \cdot p_1(x) ] - \EXP_{x \sim D}[ f^*(x) \cdot p_1(x)]} \leq \frac{64}{c^2} \sqrt{\eta(\beta_{\eta}
+ \beta_{\epsilon}) }+ \frac{\epsilon}{2}.
\end{align}

To ease notation, write $S' := S'_{l}$, $L = S \backslash S'$, $E = S' \backslash S$. Let $p_1$ be an arbitrary polynomial in $\sparselinearclass$. As $S' = S \cup E \backslash L$, it is easy to see that
\begin{align*}
&\ \abs{S'} \cdot \abs{ \EXP_{(x, y) \sim \bar{S'}}[y \cdot p_1(x) ] - \EXP_{(x, y) \sim \bar{S}}[y \cdot p_1(x) ]} \\
=&\ \abs{ (\abs{S'} - \abs{S}) \EXP_{\bar{S}}[y \cdot p_1(x)] + \abs{L} \cdot \EXP_{\bar{L}}[y \cdot p_1(x)] - \abs{E} \cdot \EXP_{\bar{E}}[y \cdot p_1(x)] }\\
\leq&\ \abs{ \abs{S'} - \abs{S} } \cdot \abs{ \EXP_{\bar{S}}[y \cdot p_1(x)] } + \abs{L} \cdot \abs{ \EXP_{\bar{L}}[y \cdot p_1(x)] } + \abs{E} \cdot \abs{ \EXP_{\bar{E}}[y \cdot p_1(x)] }.
\end{align*}
Note that the Cauchy–Schwarz inequality states that $\EXP[y \cdot p_1(x)] \leq \sqrt{\EXP[y^2]} \cdot \sqrt{\EXP[p_1^2(x)]} = \sqrt{\EXP[p_1^2(x)]}$ where the last step follows since $y \in \{-1, 1\}$. Therefore, continuing the above inequality, we have
\begin{align}\label{eq:tmp-1}
&\ \abs{S'} \cdot \abs{ \EXP_{(x, y) \sim \bar{S'}}[y \cdot p_1(x) ] - \EXP_{(x, y) \sim \bar{S}}[y \cdot p_1(x) ]} \notag \\
\leq&\ \abs{ \abs{S'} - \abs{S} } \cdot \sqrt{\EXP[p_1^2(S)]} + \abs{L} \cdot \sqrt{ \EXP[ p_1^2(L)] } + \abs{E} \cdot \sqrt{ \EXP[ p_1^2(E)] } \notag \\
\leq& \abs{ \abs{S'} - \abs{S} } \cdot \sqrt{ 1 + 2 \beta_{\delta_{\gamma}} + \epsilon } + \sqrt{  \abs{L} \cdot  \abs{S} \cdot \big( 12 \beta_{\eta} + 4 \eta  + \epsilon \big) }  + \sqrt{ \frac{6}{c} \abs{E} \cdot \abs{S'} \big( \kappa + \beta_{\eta} + \beta_{\delta_{\gamma}} + \eta + \epsilon \big) }
\end{align}
where in the last step we applied Lemma~\ref{lem:var-S}, Lemma~\ref{lem:var-L}, Lemma~\ref{lem:var-E-small-fro}, and denoted $\beta_{\delta_{\gamma}} = \beta(2\delta_{\gamma}, d, \sqrt{2})$ and $\beta_{\eta} = \beta(\eta, d, \sqrt{2})$.

On the other hand, \eqref{eq:S-S'} implies  
\begin{equation*}
\abs{ \abs{S'} - \abs{S} } \leq \frac{2\eta}{1-2\eta} \abs{S'} \leq \frac{\eta}{c} \abs{S'}
\end{equation*}
for $\eta \in [0, \frac12 - c]$. We also have the following estimates: $\max\{\abs{E}, \abs{L} \} \leq \eta \abs{S} \leq \frac{\eta}{1-2\eta} \abs{S'} \leq \frac{\eta}{2c} \abs{S'}$. Plugging these into \eqref{eq:tmp-1}, we have
\begin{equation}\label{eq:tmp-2}
\abs{ \EXP_{(x, y) \sim \bar{S'}}[y \cdot p_1(x) ] - \EXP_{(x, y) \sim \bar{S}}[y \cdot p_1(x) ]} 
\leq \frac{1}{c} \Big[ \eta \sqrt{ 1+ \beta_{\delta_{\gamma}} + \epsilon}  + 4\sqrt{\eta (\kappa + \beta_{\delta_{\gamma}} + \beta_{\eta} + \eta + \epsilon)} \Big].
\end{equation}
On the other hand, we note that in view of Part~\ref{con:p1-uniform-convergence} of Definition~\ref{def:good}, we have
\begin{equation}\label{eq:tmp-3}
\abs{ \EXP_{(x, y) \sim \bar{S}}[y \cdot p_1(x) ] - \EXP_{x\sim D}[f^*(x) p_1(x)]} = \abs{ \EXP_{x \sim S}[f^*(x)  p_1(x) ] - \EXP_{x\sim D}[f^*(x) p_1(x)] } \leq \alpha_1',
\end{equation}
where the first step follows from the condition that $f^*(\cdot)$ is the underlying PTF and $\bar{S}$ is an uncorrupted sample set (which implies $y = f^*(x)$ for any $(x, y) \in \bar{S}$). By applying triangle inequality on \eqref{eq:tmp-2} and \eqref{eq:tmp-3}, we have
\begin{equation}
\abs{\EXP_{(x, y) \sim \bar{S}' }[y \cdot p_1(x) ] - \EXP_{x \sim D}[ f^*(x) \cdot p_1(x)]} \leq \frac{4}{c} \Big[ \eta \sqrt{ 1+ \beta_{\delta_{\gamma}} + \epsilon}  + \sqrt{\eta (\kappa + \beta_{\delta_{\gamma}} + \beta_{\eta} + \eta + \epsilon)} \Big] + \alpha_1'.
\end{equation}
Now recall that  $\alpha_1' = \epsilon/6$, $\delta_{\gamma}$ is such that $\beta_{\delta_{\gamma}} \leq \beta_{\epsilon}$, $\eta \leq \beta_{\eta}$, and $\kappa \leq \frac{14}{c^2}( \beta_{\eta} + \epsilon)$. Thus, by rearrangement, we have
\begin{align*}
&\ \abs{\EXP_{(x, y) \sim \bar{S}' }[y \cdot p_1(x) ] - \EXP_{x \sim D}[ f^*(x) \cdot p_1(x)]} \\
\leq&\ \frac{16}{c^2} \Big[ \eta \sqrt{ 1+ \beta_{\epsilon} + \epsilon}  + \sqrt{\eta ( \beta_{\eta} + \beta_{\epsilon} + \epsilon)} \Big] + \frac{\epsilon}{6}\\
\stackrel{\zeta_1}{\leq}& \frac{32}{c^2} \sqrt{ \eta( \eta + \eta \beta_{\epsilon} + \eta \epsilon + \beta_{\eta} + \beta_{\epsilon} + \epsilon ) } + \frac{\epsilon}{6}\\
\stackrel{\zeta_2}{\leq}& \frac{64}{c^2} \sqrt{\eta(\beta_{\eta} + \beta_{\epsilon}) }+ \frac{\epsilon}{6},
\end{align*}
where in $\zeta_1$ we used the elementary inequality $\sqrt{a} + \sqrt{b} \leq 2\sqrt{a+b}$, and in $\zeta_2$ we used the fact that $\eta \leq \beta_{\eta}$, $\eta \epsilon < \epsilon \leq \beta_{\epsilon}$. This proves \eqref{eq:tmp:first-eq} since the above holds for any $p_1 \in \sparselinearclass$.

Now we note that for any $p_1 \in \sparselinearclass$, it can be represented as $p_1(x) = \inner{v}{m(x)}$ with $\twonorm{v} = 1$ and $\zeronorm{v} \leq 2k$. In this way, we get
\begin{align*}
&\ \abs{\EXP_{(x, y) \sim \bar{S}' }[y \cdot p_1(x) ] - \EXP_{x \sim D}[ f^*(x) \cdot p_1(x)]} \\
=&\ \abs{\EXP_{(x, y) \sim \bar{S}' }[y \cdot \inner{v}{m(x)} ] - \EXP_{x \sim D}[ f^*(x) \cdot \inner{v}{m(x)}]}\\
=&\ \abs{\inner{v}{\EXP_{(x, y) \sim \bar{S}' }[y \cdot {m(x)} ] } - \inner{v}{\EXP_{x \sim D}[ f^*(x) \cdot {m(x)}]} }\\
=&\ \abs{\inner{v}{\EXP_{(x, y) \sim \bar{S}' }[y \cdot {m(x)} ] - \Chow_{f^*}}  }.
\end{align*}
Using Lemma~\ref{lem:inner-l2} completes the proof.
\end{proof}

\subsection{Proof of Theorem~\ref{thm:main-chow}}

%\begin{theorem}\label{thm:main-chow}
%The following holds for Algorithm~\ref{alg:main}. Given any target error rate $\epsilon \in (0, 1)$, failure probability $\delta \in (0, 1)$, a sample set $\bar{S}'$ of size $C \cdot \frac{d^{5d} K^{4d}}{\epsilon^2} \log^{5d}\big(\frac{nd}{\epsilon\delta}\big)$ corrupted by the nasty adversary at a noise rate $\eta \in [0, \frac12 - c]$ for some constant $c \in (0, \frac12]$, Algorithm~\ref{alg:main} runs in at most $l_{\max} = \frac{4\eta k }{\epsilon} \cdot \big( C_1 d \cdot \log\frac{nd}{\epsilon\delta} \big)^d + 1$ phases, and outputs a $k$-sparse vector $u$ such that with probability at least $1-\delta$, 
%\begin{equation*}
%\twonorm{ u - \Chow_{f^*}} \leq \frac{192}{c^2} \sqrt{\eta(\beta_{\eta} + \beta_{\epsilon}) }+ \frac{\epsilon}{2},
%\end{equation*}
%where $\beta_{\eta} = 2 \big( c_0 \log\frac{1}{\eta} + c_0 d \big)^d \cdot \eta$ and $\beta_{\epsilon} = 2 \big( c_0 \log\frac{1}{\epsilon} + c_0 d \big)^d \cdot \epsilon$.
%In addition, Algorithm~\ref{alg:main} runs in $O(\poly((nd)^d, 1/\epsilon))$ time.
%\end{theorem}
\begin{proof}[Proof of Theorem~\ref{thm:main-chow}]
Let $\bar{S}$ be the uncorrupted sample set with the same size as $\bar{S}'$. Observe that by Proposition~\ref{prop:sample-comp}, $S$ is a good set and $\Delta(S, S') \leq 2 \eta$. We show by induction the progress of filtering, which will imply that within $l_{\max}$ phases, Algorithm~\ref{alg:main} must terminate. 

Suppose that the algorithm returns at some phase $\bar{l} \geq 1$, i.e. $\fronorm{(\Sigma - I)_U} > \kappa$ for all $1 \leq l < \bar{l}$. We show by induction that the two invariants hold: $\Delta(S, S'_l) \leq 2\eta$ and $\Delta(S, S'_{l+1}) \leq \Delta(S, S'_l) - \frac{\epsilon}{2k\gamma^2}$.

\medskip
\noindent
{\bfseries Base case: $l = 1$.} Note that since $S$ is a good set, $S_{|\Xgamma} = S$. Thus, no samples in $S$ are pruned in Step~\ref{step:main-prune} of Algorithm~\ref{alg:main}. Therefore, we have $\Delta(S, S'_1) \leq \Delta(S, S') \leq 2\eta$. In addition, we have $\fronorm{(\Sigma - I)_U} > \kappa$. Thus, Theorem~\ref{thm:large-fro} tells us that
\begin{equation}
\Delta(S, S'_2) \leq \Delta(S, S'_1) - \frac{\epsilon}{2k \gamma^2}.
\end{equation}
In particular, the above implies that $\Delta(S, S'_2) \leq 2\eta$.

\medskip
\noindent
{\bfseries Induction.} Now suppose that $\Delta(S, S'_l) \leq 2\eta$. Then applying Theorem~\ref{thm:large-fro} gives us
\begin{equation}
\Delta(S, S'_{l+1}) \leq \Delta(S, S'_{l}) - \frac{\epsilon}{2k \gamma^2},
\end{equation}
and in particular, $\Delta(S, S'_{l+1}) \leq 2\eta$.

Therefore, by telescoping, we obtain that
\begin{equation}
\Delta(S, S'_{\bar{l}}) \leq \Delta(S, S'_1) - \frac{(\bar{l}-1) \cdot \epsilon}{2k\gamma^2} \leq 2\eta - \frac{(\bar{l}-1) \cdot \epsilon}{2k\gamma^2}.
\end{equation}
On the other hand, the symmetric difference $\Delta(S, S'_{\bar{l}})$ is always non-negative. This implies that
\begin{equation}
\bar{l} \leq \frac{4\eta k \gamma^2}{\epsilon} + 1 = \frac{4\eta k }{\epsilon} \cdot \big( C_1 d \cdot \log\frac{nd}{\epsilon\delta} \big)^d + 1,
\end{equation}
where we realized the setting of $\gamma$ in the second step.

Now we characterize the output of the algorithm. In fact, by Theorem~\ref{thm:small-fro}, we have
\begin{equation*}
\twonorm{\hardthr_k\big(\EXP_{(x, y) \sim \bar{S}'_{\bar{l}} }[y \cdot {m(x)} ] \big) - \Chow_{f^*}} \leq \frac{192}{c^2} \sqrt{\eta(\beta_{\eta} + \beta_{\epsilon}) }+ \frac{\epsilon}{2}.
\end{equation*}
The proof is complete  by noting that $\hardthr_k\big(\EXP_{(x, y) \sim \bar{S}'_{\bar{l}} }[y \cdot {m(x)} ] \big)$ is the output of the algorithm.
\end{proof}

%We will pick a specific $v$ to complete the proof. Let $u = \hardthr_k \big( \EXP_{(x, y) \sim \bar{S}' }[y \cdot {m(x)} ] \big)$, which is the output of Algorithm~\ref{alg:main}. Let $\Lambda_u$ be the support set of  $u$ which naturally satisfies $\abs{\Lambda_u} \leq k$. Recall that by Lemma~\ref{lem:K-k}, $\Chow_{f^*}$ is also a $k$-sparse vector. Thus, the vector $v = \frac{u - \Chow_{f^*}}{\twonorm{u - \Chow_{f^*}}}$ has unit $\ell_2$-norm and is $2k$-sparse. Plugging it into the above and noting that $\inner{\hardthr_k(u)}{u} = \inner{\hardthr_k(u)}{\hardthr_k(u)}$, we get
%\begin{equation}
%\abs{\EXP_{(x, y) \sim \bar{S}' }[y \cdot p_1(x) ] - \EXP_{x \sim D}[ f^*(x) \cdot p_1(x)]} = \twonorm{\hardthr_k(u) - \Chow_{f^*}}.
%\end{equation}
%This completes the proof.

\subsection{Auxiliary results}

\begin{lemma}\label{lem:var-S'-small-fro}
If $\fronorm{(\Sigma-I)_U} \leq \kappa$, then
\begin{equation*}
\sup_{p_1 \in \sparselinearclass} \EXP[p_1^2(S')] \leq \kappa + 1.
\end{equation*}
\end{lemma}
\begin{proof}
For any $p_1 \in \sparselinearclass$, we can write it as $p_1(x) = v\cdot m(x)$ where $v$ is $2k$-sparse and $\twonorm{v}=1$. Denote by $J$ the support set of $v$. Then,
\begin{align*}
\EXP[p_1^2(S')] &= \EXP[v\trans m(S') m\trans(S') v] \\
&= v\trans \Sigma_{J \times J} v = v\trans (\Sigma - I)_{J \times J} v + v\trans v \leq \fronorm{v v\trans} \cdot \fronorm{ (\Sigma-I)_{J\times J}} + 1.
\end{align*}
Since $\zeronorm{v} \leq 2k$, we know that $J\times J$ has $2k$ diagonal entries and $4k^2 - 2k$ off-diagonal symmetric entries. This implies $\fronorm{(\Sigma-I)_{J\times J}} \leq \fronorm{(\Sigma-I)_U} \leq \kappa$. Now using $\fronorm{vv\trans} = \twonorm{v}^2 = 1$ completes the proof.
\end{proof}

\begin{lemma}\label{lem:var-S}
Assume that $S$ is a good set. We have
\begin{equation*}
\sup_{p_1 \in \sparselinearclass}\abs{\EXP[p_1^2(S)] - 1} \leq 2\cdot \beta(2\delta_{\gamma}, d, \sqrt{2}) + \alpha_1 \gamma_1^2.
\end{equation*}
In particular, as we set $\alpha_1 \leq \frac{\epsilon}{\gamma_1^2} $, we have 
\begin{equation*}
\sup_{p_1 \in \sparselinearclass}\abs{\EXP[p_1^2(S)] - 1} \leq 2\cdot \beta(2\delta_{\gamma}, d, \sqrt{2}) + \epsilon.
\end{equation*}
\end{lemma}
\begin{proof}

By Fact~\ref{fact:exp-by-prob},
\begin{align*}
\abs{ \EXP[p_1^2(D)] - \EXP[ p_1^2(D_{|X_{\gamma}}) ] } &= \abs{  \int_{0}^{\infty} 2t \big[  \Pr( \abs{ p_1(D)} > t) - \Pr( \lvert p_1(D_{|X_{\gamma}} ) \rvert > t) \big] \dif t }\\
&\leq \int_{0}^{\infty} 2t \min\{ 2\delta_{\gamma}, \Pr( \abs{p_1(D)} \geq t) \dif t \\
&\leq 2 \beta(2\delta_{\gamma}, d, \sqrt{2}),
\end{align*}
where the second step follows from Lemma~\ref{lem:tv-distance} and the last step follows from Lemma~\ref{lem:beta}.

On the other hand, by Part~\ref{con:p1(S-gamma)=p1(D-gamma)} of Definition~\ref{def:good}, we have
\begin{align*}
\abs{ \EXP[p_1^2(S_{|\Xgamma})] - \EXP[ p_1^2(D_{|X_{\gamma}}) ] } &= \abs{  \int_{0}^{\infty} 2t \big[  \Pr( \lvert p_1(S_{|\Xgamma}) \rvert > t) - \Pr( \lvert p_1(D_{|X_{\gamma}} ) \rvert > t) \big] \dif t }\\
&\leq \int_{0}^{\gamma_1} 2t \alpha_1 \dif t\\
&= \alpha_1 \gamma_1^2,
\end{align*}
where the inequality follows since $\abs{ p_1(x) } \leq \gamma_1$ for all $x \in \Xgamma$ (Lemma~\ref{lem:p1-properties}).
By triangle inequality, the fact that $\EXP[p_1^2(D)] = 1$ (Lemma~\ref{lem:p1-properties}), and $S_{|\Xgamma} = S$ ($S$ is a good set), we complete the proof.
\end{proof}

\begin{lemma}\label{lem:var-L}
Assume that $S$ is a good set and $\Delta(S, S') \leq 2\eta$.
Let $L = S \backslash S'$. We have
\begin{equation*}
\sup_{p_1 \in \sparselinearclass} (\abs{L} / \abs{S}) \cdot \EXP[p_1^2(L)] \leq 12\cdot \beta(\eta, d, \sqrt{2}) + 4\eta + \alpha_1 \gamma_1^2.
\end{equation*}
In particular, as we set $\alpha_1 \leq \frac{\epsilon}{\gamma_1^2}$, we have
\begin{equation*}
\sup_{p_1 \in \sparselinearclass} (\abs{L} / \abs{S}) \cdot \EXP[p_1^2(L)] \leq 12\cdot   \beta(\eta, d, \sqrt{2}) + 4 \eta + \epsilon.
\end{equation*}
\end{lemma}
\begin{proof}
We will use Lemma~\ref{lem:general-contribution} to show the result. Since $S$ is a good set, we have $S_{|\Xgamma} = S$. We have $\abs{L} \leq 2\eta \abs{S} = \abs{S_{|\Xgamma}}$ since $\Delta(S, S') \leq 2\eta$. By Lemma~\ref{lem:p1-properties} and Part~\ref{con:p1(S)=p1(D)}

in that lemma, we set $S_0 = S$, $S_1 = L$, $\omega_1 = 2$, $\epsilon_0 = \eta$. By Lemma~\ref{lem:p1-properties}, we set $Q_d(t) = e^{- (t/\sqrt{2})^{2/d} / c_0}$, and $t_0 = \sqrt{2}$. Lemma~\ref{lem:p1-properties} tells that we can set $\omega_2 = 1$. In this way, by Corollary~\ref{coro:beta-p1}, we set $\beta_0 = \beta(\eta, d, \sqrt{2})$. By Lemma~\ref{lem:p1-properties}, we can set $\gamma_0 = \gamma_1$. Therefore, we obtain the desired bound.
\end{proof}

\begin{lemma}\label{lem:var-E-small-fro}
Assume that $S$ is a good set, $\Delta(S, S') \leq 2\eta$, and $\fronorm{(\Sigma-I)_U} \leq \kappa$. We have
\begin{equation*}
\sup_{p_1 \in \sparselinearclass} \abs{E} \cdot \EXP[p_1^2(E)] \leq \abs{S'} \cdot \frac{6}{c} \big(\kappa + \beta(\eta, d, \sqrt{2}) + \beta(2\delta_{\gamma}, d, \sqrt{2}) + \eta + \epsilon \big).
\end{equation*}
\end{lemma}
\begin{proof}
Recall $S' = S \cup E \backslash L$. By algebraic calculation, we have
\begin{align*}
\abs{E} \cdot \EXP[p_1^2(E)] &= \abs{S'} \cdot \EXP[p_1^2(S')] + \abs{L} \cdot \EXP[p_1^2(L)] - \abs{S} \cdot \EXP[p_1^2(S)] \\
&\leq \abs{S'} \cdot (\kappa + 1) +  \abs{S} \cdot \big( 12 \beta(\eta, d, \sqrt{2}) + 4 \eta  + \epsilon \big) -\abs{S} \cdot \big( 1 - 2 \cdot \beta(2\delta_{\gamma}, d, \sqrt{2}) -\epsilon \big) \\
&\leq  \abs{S'} \cdot \kappa + 12 \abs{S} \big( \beta(\eta, d, \sqrt{2}) +  \beta(2\delta_{\gamma}, d, \sqrt{2}) + \eta + \epsilon \big) 
\end{align*}
where we applied Lemma~\ref{lem:var-S'-small-fro}, Lemma~\ref{lem:var-S} and Lemma~\ref{lem:var-L} in the first inequality and the fact $\abs{S'} \leq \abs{S}$ in the last step. Since $\Delta(S, S') \leq 2\eta$, in view of \eqref{eq:S-S'}, we have $\abs{S} \leq \frac{1}{1 - 2\eta} \abs{S'} \leq \frac{1}{2c} \abs{S'}$. Rearranging gives the desired result.
\end{proof}

The following lemma is similar to Lemma~\ref{lem:general-contribution-2}, but we upper bound the expectation of the square of a polynomial.
\begin{lemma}\label{lem:general-contribution}
Let $c_0 > 0$ be an absolute constant.
Let $S_0$ be a set of instances in $\Rn$ and $S_1 \subset S_0$, with $\abs{S_1} \leq \omega_1 \tau \abs{S_0}$ for $\omega_1, \tau > 0$. Let $p$ be such that $\Pr( \abs{p(S_0)} > t ) \leq \omega_2 \cdot Q_d(t) + \alpha_0$ for all $t \geq t_0$, where $\omega_2, Q_d(t), \alpha_0 > 0$. Assume $\max_{x \in S_0} \abs{p(x)} \leq \gamma_0$. Further assume that $\int_{0}^{\infty} t \min\{\tau, Q_d(t)\} \dif t \leq \beta_0$. Then
\begin{equation*}
( \abs{S_1} / \abs{S_0} ) \cdot \EXP[p^2(S_1)] \leq \omega_1 t_0^2 \cdot \tau + 2(\omega_1+1)(\omega_2+1) \beta_0  + \alpha_0\gamma_0^2.
\end{equation*}
\end{lemma}
\begin{proof}
Since $S_1 \subset S_0$, we have $\abs{S_1} \cdot \Pr( \abs{ p(S_1) } > t) \leq \abs{S_0} \cdot \Pr( \abs{ p(S_0) } > t)$. Thus,
\begin{equation}\label{eq:tmp-prob-2}
 \Pr( \abs{ p(S_1) } > t) \leq \min\Big\{ 1, \frac{\abs{S_0}}{\abs{S_1}} \cdot \Pr( \abs{ p(S_0) } > t) \Big\}.
\end{equation}
By Fact~\ref{fact:exp-by-prob}, we have
\begin{align*}
( \abs{S_1} / \abs{S_0} ) \cdot \EXP[p^2(S_1)] &\leq \int_{0}^{\infty} 2t ( \abs{S_1} / \abs{S_0} ) \Pr( \abs{ p(S_1) } > t) \dif t\\
&\stackrel{\zeta_1}{=} \int_{0}^{\gamma_0}  2t ( \abs{S_1} / \abs{S_0} ) \Pr( \abs{ p(S_1) } > t) \dif t\\
&\stackrel{\zeta_2}{\leq} \int_{0}^{\gamma_0} 2t \min\Big\{ \abs{S_1} / \abs{S_0},  \Pr( \abs{ p(S_0) } > t) \Big\} \dif t \\
&\stackrel{\zeta_3}{\leq} \int_{0}^{\gamma_0} 2t \min\Big\{ \omega_1 \tau,  \Pr( \abs{ p(S_0) } > t) \Big\} \dif t\\
&\stackrel{\zeta_4}{\leq} \int_{0}^{t_0} 2t \min\{ \omega_1 \tau, 1\} \dif t + \int_{t_0}^{\gamma_0} 2t \min\{ \omega_1 \tau, \omega_2 \cdot Q_d(t) + \alpha_0 \} \dif t\\
&\stackrel{\zeta_5}{\leq} \int_{0}^{t_0} 2 \omega_1 \tau t \dif t + \int_{t_0}^{\gamma_0} 2t \min\{ \omega_1 \tau, \omega_2 \cdot Q_d(t)  \} \dif t + \int_{t_0}^{\gamma_0} 2t  \alpha_0 \dif t \\
&\stackrel{\zeta_6}{\leq} \omega_1 t_0^2 \cdot \tau + 2(\omega_1+1)(\omega_2+1) \int_{t_0}^{\gamma_0} t \min\{  \tau,   Q_d(t)  \} \dif t  + \alpha_0(\gamma_0^2 - t_0^2)\\
&\stackrel{\zeta_7}{\leq} \omega_1 t_0^2 \cdot \tau + 2(\omega_1+1)(\omega_2+1) \beta_0  + \alpha_0\gamma_0^2.
\end{align*}
In the above, $\zeta_1$ follows from the condition that $p(x) \leq \gamma_0$ for all $x \in S_1$, $\zeta_2$ follows from \eqref{eq:tmp-prob-2}, $\zeta_3$ uses the condition $\abs{S_1} \leq \omega_1 \tau \abs{S_0}$, $\zeta_4$ uses the condition of the tail bound of $p(S_0)$ when $t \geq t_0$, $\zeta_5$ applies elementary facts that $\min\{ \omega_1 \tau, 1\} \leq \omega_1 \tau$ and $\min\{a, b + c\} \leq \min\{a, b\} + c $ for any $c > 0$, $\zeta_6$ uses the fact that both $\frac{\omega_1}{(\omega_1+1)(\omega_2+1)}$ and $\frac{\omega_1}{(\omega_1+1)(\omega_2+1)}$ are less than $1$ for positive $\omega_1$ and $\omega_2$, and $\zeta_7$ applies the condition on the integral and uses the fact that both $\tau$ and $Q_d(t)$ are positive.
\end{proof}

\begin{lemma}\label{lem:tv-distance}
Suppose $\delta_{\gamma} \leq 1/2$.
The following holds for any function $p$:
\begin{equation*}
\abs{ \Pr( \lvert p(D_{|X_{\gamma}}) \rvert \geq t ) - \Pr( \abs{ p(D) } \geq t) } \leq \min \big\{ 2\delta_{\gamma}, \Pr(\abs{p(D)} \geq t ) \big\}.
\end{equation*}
\end{lemma}
\begin{proof}
Lemma~\ref{lem:gamma} tells that $\Pr_{x \sim D}(x \notin X_{\gamma}) \leq \delta_{\gamma}$.
Observe that
\begin{equation*}
\Pr( \lvert p(D_{|X_{\gamma}}) \rvert \geq t ) \leq \frac{ \Pr ( \abs{ p(D) } \geq t) }{ \Pr_{x\sim D}(x \in X_{\gamma} )} \leq \frac{1}{1- \delta_{\gamma}} \Pr( \abs{ p(D) } \geq t) \leq 2 \Pr( \abs{ p(D) } \geq t),
\end{equation*}
implying
\begin{equation}
\abs{ \Pr( \lvert p(D_{|X_{\gamma}}) \rvert \geq t ) - \Pr( \abs{ p(D) } \geq t) } \leq \Pr( \abs{ p(D) } \geq t).
\end{equation}
On the other hand,  for any event $A$,
\begin{equation*}
\Pr_{D}(A) = \Pr_D(A \mid x \in X_{\gamma}) \cdot \Pr_{D}(x \in X_{\gamma}) + \Pr_D( A \mid x \notin X_{\gamma}) \cdot \Pr_D( x \notin X_{\gamma}).
\end{equation*}
This implies
\begin{align*}
&\ \abs{ \Pr_D(A) - \Pr_D(A \mid x \in X_{\gamma} ) } \\
=&\ \abs{ \Pr_{D}(A \mid x \in X_{\gamma}) \cdot ( \Pr_D(x \in X_{\gamma}) - 1) + \Pr_D(A \mid x \notin X_{\gamma}) \cdot \Pr_D( x \notin X_{\gamma}) }\\
=&\  \abs{ - \Pr_{D}(A \mid x \in X_{\gamma}) \cdot \Pr_D(x \notin X_{\gamma})  + \Pr_D(A \mid x \notin X_{\gamma}) \cdot \Pr_D( x \notin X_{\gamma}) }\\
\leq&\ 2 \Pr_D(x \notin X_{\gamma} )\\
\leq&\ 2\delta_{\gamma}.
\end{align*}
This completes the proof.
\end{proof}

\begin{lemma}\label{lem:inner-l2}
For any vector $w$ and any $k$-sparse vector $u$, if $\sup_{v: \twonorm{v}=1, \zeronorm{v} \leq 2k} \abs{\inner{v}{w - u}} \leq \epsilon$, then $\twonorm{\hardthr_k(w) - u} \leq 3\epsilon$.
\end{lemma}
\begin{proof}
Let $\Lambda_0$ be the support set of $\hardthr_k(w)$, $\Lambda_1 = \supp(u) \backslash \Lambda_0$, $\Lambda_2 = \supp(u) \cap \Lambda_0$, $\Lambda_3 = \Lambda_0 \backslash \supp(u)$. Therefore, we can decompose
\begin{equation}\label{eq:hard-0}
\twonorm{\hardthr_k(w) - u}^2 = \twonorm{u_{\Lambda_1}}^2 + \twonorm{(w-u)_{\Lambda_2}}^2 + \twonorm{w_{\Lambda_3}}^2.
\end{equation}
Note that by choosing $v = (w - u)_{\Lambda_2 \cup \Lambda_3} / \twonorm{(w - u)_{\Lambda_2 \cup \Lambda_3}}$, we get
\begin{equation}\label{eq:tmp-hard-1}
\twonorm{(w-u)_{\Lambda_2}}^2 + \twonorm{w_{\Lambda_3}}^2 \leq \epsilon^2.
\end{equation}
On the other side, observe that
\begin{equation}\label{eq:tmp-hard-2}
\twonorm{u_{\Lambda_1}} \leq \twonorm{(u-w)_{\Lambda_1}} + \twonorm{w_{\Lambda_1}}
\end{equation}
by triangle inequality. Since $\Lambda_3$ is a subset of $\Lambda_0$, the index set of the $k$ largest elements of $w$, while $\Lambda_1 \cap \Lambda_0 = \emptyset$, we know that elements of $w$ in $\Lambda_1$ are less than those in $\Lambda_3$. This combined with the fact that $\abs{\Lambda_1} = \abs{\Lambda_3}$ implies that
\begin{equation}\label{eq:tmp-hard-3}
\twonorm{w_{\Lambda_1}} \leq \twonorm{w_{\Lambda_3}} \leq \epsilon.
\end{equation}
where the second step follows from \eqref{eq:tmp-hard-1}. In order to upper bound $\twonorm{(u-w)_{\Lambda_1}}$, we can pick $v = (u-w)_{\Lambda_1} / \twonorm{(u-w)_{\Lambda_1}}$ and get
\begin{equation}\label{eq:tmp-hard-4}
\twonorm{(u-w)_{\Lambda_1}} \leq \epsilon.
\end{equation}
Plugging \eqref{eq:tmp-hard-3} and \eqref{eq:tmp-hard-4} into \eqref{eq:tmp-hard-2} gives
\begin{equation*}
\twonorm{u_{\Lambda_1}} \leq 2 \epsilon.
\end{equation*}
This in conjunction with \eqref{eq:tmp-hard-1} and \eqref{eq:hard-0} gives $\twonorm{\hardthr_k(w) - u} \leq \sqrt{5}\epsilon \leq 3\epsilon$.
\end{proof}

\section{Proof of Theorem~\ref{thm:main}}\label{sec:app:PAC}

%\begin{theorem}[Main results]\label{thm:main}
%There exists an algorithm such that the following holds. Given any target error rate $\epsilon \in (0, 1)$, failure probability $\delta \in (0, 1)$, a sample set $\bar{S}'$ of size $C \cdot \frac{d^{5d} K^{4d}}{\epsilon^2} \log^{5d}\big(\frac{nd}{\epsilon\delta}\big)$ corrupted by the nasty adversary at a noise rate $\eta \in [0, \frac12 - c]$ for some constant $c \in (0, \frac12]$, the algorithm outputs a PTF $\hat{f}$ such that with probability at least $1-\delta$, 
%\begin{equation*}
%\Pr_{x\sim D}( \hat{f}(x) \neq f^*(x)) \leq c_1 \cdot d \cdot \Big( \frac{192}{c^2} \sqrt{\eta(\beta_{\eta} + \beta_{\epsilon}) }+ \frac{\epsilon}{2} \Big)^{\frac{1}{d+1}},
%\end{equation*}
%where $\beta_{\eta} = 2 \big( c_0 \log\frac{1}{\eta} + c_0 d \big)^d \cdot \eta$ and $\beta_{\epsilon} = 2 \big( c_0 \log\frac{1}{\epsilon} + c_0 d \big)^d \cdot \epsilon$. In particular, for any given $\epsilon_0 \in (0, 1)$,
%\begin{equation*}
%\Pr_{x\sim D}( \hat{f}(x) \neq f^*(x)) \leq \epsilon_0,
%\end{equation*}
%provided that $\epsilon = \frac{\epsilon_0^{d+1}}{c_2 \cdot  d^{2d}}$ and $\eta \leq \frac12 \epsilon_0$ for some large enough constant $c_2 > 0$. In addition, the algorithm runs in $O(\poly(n^d, 1/\epsilon))$ time.
%\end{theorem}
\begin{proof}
The sample complexity and the first equation are an immediate result from Theorem~\ref{thm:main-chow} and Lemma~\ref{lem:chow-to-pac}. The second equation follows from algebraic calculation.
\end{proof}

\begin{lemma}[\citet{diakonikolas2018learning}]\label{lem:chow-to-pac}
Suppose $D$ is $\calN(0, I_{n\times n})$. There is an algorithm that takes as input a vector $u \in \R^{(n+1)^d}$ with $\twonorm{u - \Chow_{f^*}} \leq \epsilon$, runs in time ${O}(n^{d}/ \epsilon^{2})$ and outputs a degree-$d$ PTF $\hat{f}$ such that
\begin{equation*}
\Pr_{x \sim D}\big( \hat{f}(x) \neq f^*(x) \big) \leq c_1 \cdot d \cdot \epsilon^{\frac{1}{d+1}},
\end{equation*}
for some absolute constant $c_1 > 0$.
\end{lemma}
\begin{proof}
The result is a combination of Theorem~10 of \citet{DDFS14}, Lemma~3.4 and Lemma~3.5 of \citet{diakonikolas2018learning}. The only difference is that when applying Theorem~10 of \citet{DDFS14} to our setup, we can compute Chow vectors of a given function exactly since $D$ is known to be Gaussian; thus no additional samples are needed and the running time is slightly better than their original analysis. See Lemma~\ref{lem:chow-to-PBF} for clarification.
\end{proof}

We reproduce the proof of Theorem~10 of \citet{DDFS14} but tailored to our case that $D$ is Gaussian, and thus there is no need to acquire additional samples.
\begin{lemma}\label{lem:chow-to-PBF}
Let $f$ be a degree-$d$ PTF on $\Rn$.
There is an algorithm that takes as input a vector $v$ with $\twonorm{v - \Chow_f} \leq \epsilon$, and outputs a polynomial bounded function $g: \Rn \rightarrow [-1, 1]$ such that $\twonorm{\Chow_g - \Chow_f}\leq 4\epsilon$. In addition, the algorithm runs in $O(n^d/\epsilon^2)$ time.
\end{lemma}

\begin{proof}
We will iteratively construct a sequence of functions $\{g_t\} \subset \polyclass$. Let $g_0' = 0$ and $g_0 = P_1(g_0')$, where $P_1(a) = \sign(a)$ if $\abs{a} \geq 1$ and $P_1(a) = a$ otherwise. Given $g_t$, we compute $\Chow_{g_t}$. Let  
\begin{equation}
\rho := \twonorm{v - \Chow_{g_t} }.
\end{equation}

\paragraph{Case 1. $\rho \leq  3\epsilon$.}

Then 
\begin{align*}
\twonorm{\Chow_{g_t} - \Chow_f} \leq \twonorm{\Chow_{g_t} - v }  + \twonorm{v - \Chow_f} = \twonorm{ \Chow_{g_t} - v  } + \twonorm{v - \Chow_f} \leq 4 \epsilon.
\end{align*}
Thus we output $g_t$.

\paragraph{Case 2. $\rho > 3\epsilon$.}

Define
\begin{equation}
h_t'(x) = (v - \Chow_{g_t}) \cdot m(x),\quad g_{t+1}' = g_t' + \frac{1}{2}h_t',\quad g_{t+1} = P_1(g_{t+1}'). 
\end{equation}
Consider the following potential function:
\begin{equation}
E(t) = \EXP[(f - g_t) (f - 2 g_t' + g_t)].
\end{equation}

The proof of Theorem~10 of \citet{DDFS14}  implies the following two claims:

\begin{claim}\label{claim:(f-g)h-lower}
$\EXP[(f - g_t) h_t'] \geq \rho (\rho - \epsilon)$.
\end{claim}

\begin{claim}
Given any $g_t'$ and $h_t'$, let $g_t = P_1(g_t')$, $g_{t+1}' = g_t' + \frac12 h_t'$, $g_{t+1} = P_1(g_{t+1}')$. Then
$\EXP[(g_{t+1} - g_t)(2g_{t+1}' - g_t - g_{t+1})] \leq \frac12 \EXP[(h_t')^2]$.
\end{claim}
Observe that by our definition of $h_t'$, we have
\begin{equation}
\EXP[(h_t')^2] = \twonorm{v - \Chow_{g_t}}^2 = \rho^2.
\end{equation}
Therefore, the progress of $E(t)$ can be bounded as follows:
\begin{align*}
E(t+1) - E(t) &= - \EXP[(f - g_t) h_t'] + \EXP[ (g_{t+1} - g_t)(2g_{t+1}' - g_t - g_{t+1})]\\
&\leq -\rho(\rho - \epsilon) + \frac12 \rho^2\\
&\leq - \epsilon^2.
\end{align*}
In addition, we have $E(t) \geq 0$ and $E(0) = 1$. These together imply that the algorithm terminates in at most $\frac{1}{\epsilon^2}$ iterations. It is easy to see that the computational cost in each iteration is dominated by the construction of $h_t'(\cdot)$, which is $O(n^d)$. Thus, the overall running time is $O(n^d/\epsilon^2)$.
\end{proof}